\documentclass[aps,prl,twocolumn,superscriptaddress,nofootinbib,floatfix]{revtex4-2}

\usepackage[T1]{fontenc}
\usepackage[utf8]{inputenc}
\usepackage{amsmath,amssymb,mathtools,bm}
\usepackage{graphicx}
\usepackage{booktabs}
\usepackage[dvipsnames]{xcolor}
\usepackage{soul}
\usepackage[colorlinks=true,citecolor=blue,linkcolor=blue,urlcolor=blue]{hyperref}

\newcommand{\cA}{\mathcal A}
\newcommand{\cH}{\mathcal H}
\newcommand{\cM}{\mathcal M}
\newcommand{\cP}{\mathcal P}
\newcommand{\cS}{\mathcal S}
\newcommand{\cV}{\mathcal V}
\newcommand{\RR}{\mathbb R}

\newcounter{prltheorem}
\newenvironment{prltheorem}{%
  \par\smallskip\noindent\refstepcounter{prltheorem}\textbf{Theorem \theprltheorem}\ \begingroup\itshape
}{%
  \par\endgroup\smallskip
}

\newcommand{\update}[1]{\textcolor{black}{#1}}
\newcounter{prlcorollary}
\newenvironment{prlcorollary}[1]{%
  \par\smallskip\noindent\refstepcounter{prlcorollary}\textbf{Corollary \theprlcorollary\ (#1).}\ \begingroup\itshape
}{%
  \par\endgroup\smallskip
}
\newif\ifshowchanges
\showchangestrue
\ifshowchanges
  
  \newcommand{\remove}[1]{{\color{lightgray}\st{#1}}}

  \newcommand{\mremove}[1]{{\color{lightgray}#1}}
  
\else
  
  \newcommand{\remove}[1]{}

  \newcommand{\mremove}[1]{}
  
\fi

\begin{document}

\title{Algebraic Speedups for Exact Inversion of Hamiltonian Evolutions}
\author{Jizhe Lai}
\author{Mingrui Jing}
\author{Erdong Huang}
\affiliation{Thrust of Artificial Intelligence, Information Hub,\\
The Hong Kong University of Science and Technology (Guangzhou), Guangzhou 511453, China}
\affiliation{QudeLeap Research, Shanghai 200030, China}
\author{Xin Wang}
\email{felixxinwang@hkust-gz.edu.cn}
\affiliation{Thrust of Artificial Intelligence, Information Hub,\\
The Hong Kong University of Science and Technology (Guangzhou), Guangzhou 511453, China}

\begin{abstract}
Deterministic exact inversion of an arbitrary $d$-dimensional unitary requires {$\Theta(d^2)$} coherent forward calls in the worst case. We ask how this cost changes for Hamiltonian evolution $U(x)=\exp(i\sum_j x_jH_j)$ when the generators are known but the parameters are hidden. For one-parameter families with a fixed eigenbasis, we show that additive relations among the distinct eigenvalues determine the optimal query number exactly, and we construct the corresponding inversion protocol.
For general families, we prove that repeated symmetry sectors do not affect the exact query complexity and give an automatic construction for combining inverses from inequivalent active sectors. We also give a sufficient phase-alignment condition under which family-specific structure can reduce 
{the query number.}
These results establish structure-dependent bounds for reversing the unknown dynamics arising in Tavis-Cummings out-of-time-order correlator protocols, collective-spin echo verification, and passive multimode links, without requiring prior knowledge or explicit estimation of the underlying coupling strengths.
\end{abstract}

\maketitle

\paragraph*{Introduction.---}
Reversing a many-body quantum evolution is a basic operation for
probing dynamics.  Loschmidt echoes, echo-based verification, and
out-of-time-order correlator (OTOC) protocols compare a forward
evolution with a backward branch to diagnose irreversibility, control
errors, and information scrambling
~\cite{Oreshkov2015TimeReversal,Swingle2016MeasuringScrambling}.
For an evolution $U=e^{-iHt}$ {with Hamiltonian $H$}, the ideal backward branch is $U^\dagger=e^{iHt}$, formally obtained by changing $t$ to $-t$.
In an experiment, however, the hardware may 
{keep}
the interaction graph, the allowed coupling terms, and their symmetries, while the realized coupling strengths and detunings remain imperfectly calibrated or slowly drift~\cite{BaireyAradLindner2019,CastanedaWiebe2025PseudoChoi}. Knowing the form of the interactions is therefore not equivalent to having direct access to the sign-reversed Hamiltonian.

One possible route is to learn the realized Hamiltonian and then synthesize its inverse.  Such a procedure produces a finite-precision classical or quantum representation and incurs an accuracy-dependent cost ~\cite{mohseni2008quantum,Bisio2010,Zhu2024HamiltonianRecognition}.
Dynamical reversal  requires the action of the original evolution, rather than an estimate of the parameters that generate it. We consider a Hamiltonian family $H(x)=\sum_{j}x_{j}H_{j}$ generating $U(x)=e^{iH(x)}$, where the generators $H_{j}$ are known but the coefficients $x$ are hidden. 
Given access only to forward calls to $U(x)$, {we ask for the minimal number of calls required} 
for a fixed quantum circuit to implement $U(x)^\dagger$ for every allowed $x$, deterministically and exactly, without estimating $x$.
We call this minimum query number the reversing cost of the Hamiltonian family. \textit{
We therefore ask whether structured many-body Hamiltonian families can have a reversing cost that grows only polynomially with the number of qubits in the systems rather than exponentially.
} Such a circuit is a higher-order quantum transformation assembled from repeated calls to the input dynamics and fixed parameter-independent operations ~\cite{Chiribella2008Supermap,Chiribella2008Circuit, Chiribella2009Networks,Bisio2019HigherOrder,Milz2024Higher}.

\begin{figure}[t!]
\centering
\includegraphics[width=\columnwidth,trim=1.4cm 0cm 0cm 0cm, clip]{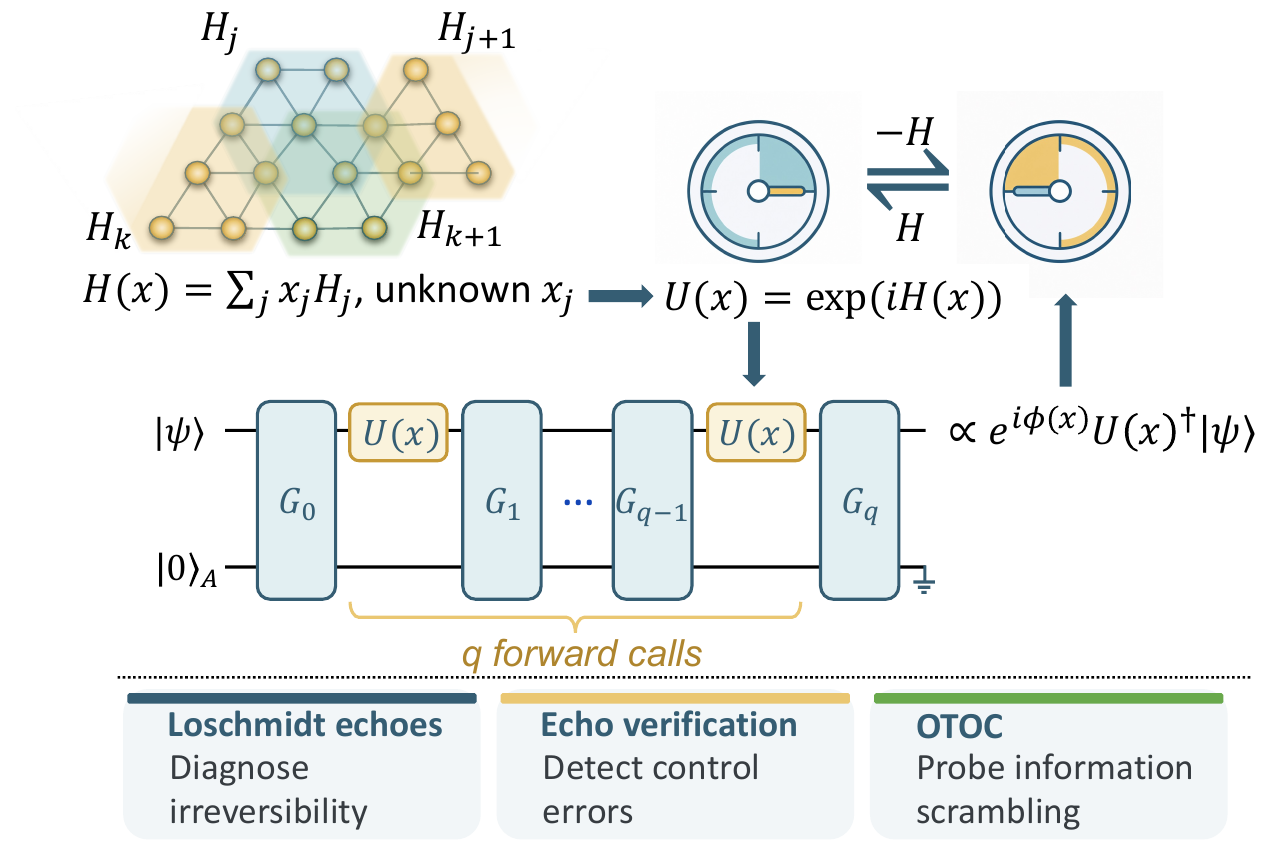}
\caption{\textbf{Structured Hamiltonian inversion.}
The Hamiltonian generators $H_j$ and their algebraic relations are known, while their coefficients $x_j$ are hidden. A fixed quantum circuit interleaves $q$ forward calls to $U(x)=e^{iH(x)}$ with parameter-independent gates and implements $U(x)^\dagger$ up to a global phase, while restoring all auxiliary registers. The resulting inverse can be used in Loschmidt echoes, echo verification, and out-of-time-order correlator protocols.}
\label{fig:promise}
\end{figure}

Deterministic exact inversion of an arbitrary $d$-dimensional unitary requires optimally $\mathcal{O}(d^2)$ queries~\cite{Yoshida2023QubitInversion,Chen2024QURA, Odake2025AnalyticalLowerBound}, and the dimension dependence persists even with constant approximation error~\cite{ChenYuZhang2025TimeReversal}, making universal inversion exponentially costly for generic many-body dynamics.  
Probabilistic, approximate, and Hamiltonian-transformation protocols provide alternative primitives~\cite{Quintino2019PRL,Quintino2019PRA,QuintinoEbler2022, Odake2024HamiltonianDynamics,Odake2025EigenvalueTransformation}, 
Lower-bound methods have also been extended to inputs promised to lie in a Lie subgroup~\cite{Odake2025AnalyticalLowerBound}. One constructive route applies when a fixed operation $V$ satisfies $V H_j V^\dagger=-H_j$ for every generator, enabling one-query inversion~\cite{MoLinWang2025Structured}.
Many physical families admit no such common sign flip, raising the question of whether weaker spectral or representation structure can still reduce the exact query number~\cite{KnillLaflammeViola2000, EtingofBook}.

In this paper, we show that algebraic structure reduces the number of queries required for exact inversion of unitaries generated by Hamiltonian families. For families of commuting Hamiltonians, an additive-sumset criterion determines the reversing cost exactly from the distinct joint-eigenvalue vectors; the cost is the optimal value of an explicit integer program. For noncommuting generators, the Wedderburn decomposition of the generated matrix algebra reduces the family to its inequivalent representations, proving that repeated symmetry sectors do not affect the query complexity. We further derive a sufficient phase-synchronization condition for combining the sector inverses into a coherent global inverse. Finally, we apply these results to the Tavis--Cummings family, the collective-spin family, and passive multimode systems, where the resulting bounds depend on the active symmetry sectors or spectral labels rather than on the full Hilbert-space dimension.

\paragraph*{Exact inversion from forward evolutions.---}
Consider a structured unitary family 
{with $x\in\RR^p$}
\begin{equation}
    U(x)=\exp\!\left(i\sum_{j=1}^{p}x_jH_j\right)
    \label{eq:family}
\end{equation}
\update{defined by $p$ linearly independent Hermitian generators $H_j$ acting on the system Hilbert space. The generators are known, whereas $x$ is unknown. The protocol may use repeated forward calls to $U(x)$, interleaved with fixed $x$-independent quantum gates. These gates may depend on the known generator family, for example through a known eigenbasis or symmetry-adapted basis.}

A deterministic $q$-query protocol, or quantum comb~\cite{Chiribella2008Circuit}, interleaves fixed quantum gates with $q$ forward calls. Following Refs.~\cite{Gavorova2024Topological,Yoshida2023QubitInversion}, we consider the \textit{clean} inversion protocols, which return every auxiliary register to its initial state. The target is $e^{i\phi(x)}U(x)^\dagger$ for every $x$, where $\phi(x)\in\RR$ is a parameter-dependent global phase. Only the query number is counted. More explicitly, let $A$ denote the auxiliary registers, initialized in $|0\rangle_A$, and set $Q_x=U(x)\otimes I_A$. For fixed quantum gates $G_0,\ldots,G_q$, {acting on both the system and auxiliary registers}, 
the protocol satisfies
\[
\begin{aligned}
&G_qQ_xG_{q-1}\cdots G_1Q_xG_0
\bigl(|\psi\rangle\otimes|0\rangle_A\bigr)\\
&\quad\qquad\qquad=e^{i\phi(x)}U(x)^\dagger|\psi\rangle\otimes |0\rangle_A
\end{aligned}
\]
for every $x$ and every input state $|\psi\rangle$. We refer to the smallest such $q$ as the \textit{reversing cost} of the Hamiltonian family.

\paragraph*{Reversing dynamics generated from commuting Hamiltonian families.---} 
Consider a Hamiltonian $H$ with known spectral decomposition $H=\sum_{j=0}^{K-1}\lambda_j\Pi_j$, where $\lambda_j$ are distinct and $\Pi_j$ is the projector onto the corresponding eigenspace. Let $\Lambda=\{\lambda_0,\ldots,\lambda_{K-1}\}$. We first consider the one-parameter family $U(x)=e^{iHx}$ with hidden evolution parameter $x$; the target is $U(x)^\dagger=e^{-iHx}$ up to a global phase. 
For $q\ge0$, let $\Sigma_q(\Lambda)$ denote {the set of} all sums of $q$ eigenvalues, with repetitions allowed:
\[
\Sigma_q(\Lambda)
=\left\{\lambda_1+\cdots+\lambda_q:\lambda_1,\ldots,\lambda_q\in\Lambda\right\},
\quad \Sigma_0(\Lambda)=\{0\}.
\]
{Routing one query through the $\lambda_j$ eigenspace contributes $e^{i\lambda_jx}$. Every branch of a $q$-query protocol passes through all $q$ calls, so each element of $\Sigma_q(\Lambda)$ collects exactly $q$ spectral labels, one per query. An input component can therefore be inverted only when its spectral label can be completed to a common total frequency.}

\begin{prltheorem}\label{thm:spectral-seed}
Let $U(x) = e^{iHx}$ be a one-parameter unitary family generated by $H$. The minimum number of forward calls required to implement $U(x)^\dagger$ up to a global phase for unknown $x\in\RR$ is
\begin{equation}
\begin{aligned}
&\kappa(\Lambda)=
\min\!\left\{q\in\mathbb Z_{\ge0}:\begin{array}{l}
\exists c\in\RR\ \text{such that}\\[-2pt]
c-\lambda\in\Sigma_q(\Lambda)\ \forall\lambda\in\Lambda
\end{array}\right\}.
\end{aligned}
\end{equation}
\end{prltheorem}

\emph{Proof sketch:---} 
{For sufficiency, suppose that $c-\lambda\in\Sigma_q(\Lambda)$ for every $\lambda\in\Lambda$. For each $\lambda$, choose $q$ spectral labels whose sum is $c-\lambda$, and route the corresponding input component through those eigenspaces. The accumulated phase is $e^{icx}e^{-i\lambda x}$, so the protocol implements $U(x)^\dagger$ up to the common phase $e^{icx}$.} The multiplicity label is stored on an auxiliary register that is restored at the end and never undergoes the unknown evolution. 

For necessity, every matrix element of a $q$-query protocol is a finite exponential polynomial whose frequencies belong to $\Sigma_q(\Lambda)$. If the protocol gives $e^{i\phi(x)}U(x)^\dagger$, 
then $e^{i\phi(x)}e^{-i\lambda x}$ has support in that set for every $\lambda$, where $e^{i\phi(x)}$ is the residual global phase. {A finite exponential sum with distinct real frequencies that has unit modulus for all $x\in\RR$ must consist of a single mode: otherwise the difference between its largest and smallest frequencies would contribute a unique nonvanishing Fourier coefficient to its squared modulus. The complete argument is given in the Supplemental Material~\cite{SupplementalMaterial}.}

{In particular, when $H$ has $K$ distinct eigenspaces, we have $\kappa(\Lambda) \le K-1$, independent of all eigenspace multiplicities~\cite{SupplementalMaterial}. 
The universal $K-1$ construction routes the query register once through every distinct eigenspace except the input eigenspace.} 

For commuting Hamiltonian families, the corresponding unitary can always be written as~\cite{HornJohnson2012MatrixAnalysis}
\[
 U(x)=\sum_{\lambda\in\Lambda}e^{i\lambda\cdot x}\Pi_\lambda,
 \qquad \Lambda\subset\RR^p.
\]
The same argument gives the multiparameter condition as in Theorem~\ref{thm:spectral-seed}, now with $c\in\RR^p$ and
$\Sigma_q(\Lambda)$ defined by vector addition. The routing construction is unchanged, while necessity follows by restricting the finite multivariate exponential polynomial to a generic line that separates its frequencies. A complete proof is given in the Supplemental Material.

\paragraph*{Reversing dynamics generated from general Hamiltonian families.---} 
For noncommuting generators, the eigenspaces of $\sum_jx_jH_j$ generally depend on $x$, so there is no fixed spectral basis in which to apply Theorem~\ref{thm:spectral-seed}. We instead use the finite-dimensional matrix algebra 
generated by the known Hamiltonian terms~\cite{KnillLaflammeViola2000,EtingofBook}. Wedderburn theory states that, after a known change of basis, $U(x)$ decomposes into a direct sum of inequivalent matrix families $U_\alpha(x)$. The $\alpha$th family acts on a $d_\alpha$-dimensional space and may appear $m_\alpha$ times, with every copy undergoing the same unknown matrix. We denote the corresponding restricted generators by $H_{j,\alpha}$. Repeated copies therefore carry no additional dependence on $x$. 

\begin{prltheorem}\label{thm:algebraic-inversion}
Let $U(x)=\exp{(iH(x))}$ be the unitary subset generated by the Hermitian operators $H(x) = \sum_j x_jH_j$ where $x = (x_j)$ is the collection of all components. 
Suppose that the $\alpha$-th block matrix family $U_\alpha(x)$ can be reversed with $q_{\alpha}$ forward calls. Then the full family can be reversed using at most
$\sum_{\alpha=1}^{r}(q_{\alpha}+1)-1$
forward calls. In particular, applying a general $d_\alpha$-dimensional unitary inverter~\cite{Chen2024QURA} to each inequivalent matrix family 
gives the constructive upper bound $\mathcal O(\sum_{\alpha=1}^{r}d_\alpha^2)$.
\end{prltheorem}

A fixed basis change stores the copy label and routes the active state through one representative; reversing this routing gives the converse direction with the same number of forward calls.
All $m_\alpha$ copies undergo the same unknown matrix $U_\alpha(x)$, and a fixed basis change can store the copy label while routing one representative through the block access of $U(x)$. The same routing works in reverse, so a $q$-query inverse exists for the full family if and only if one exists for the reduced family.

\emph{Proof sketch:---}{This upper bound follows from a phase-completion argument. Let a $q_{\alpha}$-query inverse of $U_\alpha(x)$ leave the residual phase $\rho_\alpha(x)$. Appending one further forward call cancels the inverse matrix action: the resulting $(q_\alpha+1)$-call sequence acts on sector $\alpha$ as the scalar $\rho_\alpha(x)$. For an input supported in sector $\alpha$, the protocol applies the $q_\alpha$-query inverse to the input register and, for every $\beta\neq\alpha$, runs the closed $(q_\beta+1)$-call sequence on an auxiliary register prepared in a fixed state of sector $\beta$. Each auxiliary register returns to its initial state, releasing $\rho_\beta(x)$ as a global phase. Every input sector therefore uses $\sum_\beta(q_\beta+1)-1$ calls and acquires the same phase $\prod_\beta\rho_\beta(x)$, so the sector inverses combine coherently on superpositions. All sector-dependent routing is performed by fixed gates, so the construction 
uses only ordinary calls to $U(x)$; direct access to the individual blocks is not assumed.
The universal protocol gives $q_{\alpha}=\mathcal O(d_\alpha^2)$ and hence the last bound in Theorem~\ref{thm:algebraic-inversion}.}

The automatic construction is conservative. Fewer calls can suffice when shorter scalar-returning sequences make all matrix sectors use the same total query number and acquire the same $x$-dependent phase. Matching both quantities is sufficient for coherent recombination. The scalar case is characterized exactly by Theorem~\ref{thm:spectral-seed}; the general matrix-valued condition and its circuit construction are given in the Supplemental Material~\cite{SupplementalMaterial}.

For a universal inverse of a $d_\alpha\times d_\alpha$ unitary, $q_{\alpha}+1=d_\alpha s_\alpha$ and the residual phase is $\rho_\alpha(x)=\det[U_\alpha(x)]^{s_\alpha}$ for some positive integer $s_\alpha$~\cite{Gavorova2024Topological}. Since
$\det U_\alpha(x)=\exp(i\sum_jx_j\operatorname{tr}H_{j,\alpha})$, the trace vector
$(\operatorname{tr}H_{1,\alpha},\ldots,\operatorname{tr}H_{p,\alpha})$ determines the phase carried by that matrix sector. If every restricted generator is traceless, these determinant phases disappear and only the query lengths must be matched. Formal circuit definitions, complete proofs, and the parity conditions for this remaining padding problem are given in the Supplemental Material~\cite{SupplementalMaterial}.

\paragraph*{Structure-dependent scaling in physical models.---} We now apply the results above to three types of many-body dynamics that require a deterministic inverse of the forward evolution. Table~\ref{tab:applications} compares selected structured
construction bounds with dimension-only benchmarks obtained by treating the corresponding Hilbert space as an arbitrary unitary.
Figure~\ref{fig:link-hierarchy} makes the same comparison for passive multimode links. The dimension-only curve treats the full truncated photon space as an arbitrary unitary, whereas the other curves exploit the same single-particle mode transformation experienced by every photon, together with additional spectral structure for ring and bright-mode couplings.

\par\smallskip
\begin{center}
\refstepcounter{table}\label{tab:applications}
\begin{minipage}{\columnwidth}
\footnotesize
\textbf{TABLE~\thetable.} \textbf{Structured bounds and universal benchmarks.}
Here $n$ counts emitters, spins, or optical modes in the corresponding settings. For Tavis--Cummings dynamics, $M$ and $N$ are the fixed mode number and excitation cutoff; the passive-link rows use $N\le2$. The notation $\mathcal O_{M,N}(1)$ means that the query count is bounded independently of $n$ when $M$ and $N$ are fixed.
\par\smallskip
\centering
\begin{tabular}{@{}p{0.44\columnwidth}cc@{}}
\toprule
Setting & Structured & Universal \\
\midrule
Tavis--Cummings OTOC~\cite{TiwariBanerjee2023DickeChaos}
& $\mathcal{O}_{M,N}(1)$ & $\mathcal O(n^{2N})$ \\
Collective-spin echo~\cite{Weaving2025EVCDR}
& $\mathcal{O}(n^3)$ & $\mathcal O(4^n)$ \\
Passive link, arbitrary $X$~\cite{Lu2024OpticalDescrambler}
& $\mathcal{O}(n^2)$ & $\mathcal O(n^4)$ \\
Passive link, circulant $X$
& $n$ & $\mathcal O(n^4)$ \\
\bottomrule
\end{tabular}
\end{minipage}
\end{center}
\smallskip

\emph{Tavis-Cummings family.---} We consider the Tavis--Cummings Hamiltonian studied in atom--field OTOCs~\cite{TiwariBanerjee2023DickeChaos},
\begin{equation}
H_{\rm TC}=
\omega_0J_z+\sum_{\mu=1}^{M}\omega_\mu a_\mu^\dagger a_\mu
+\sum_{\mu=1}^{M}\left(g_\mu J_+a_\mu+g_\mu^*J_-a_\mu^\dagger\right)
\label{eq:TC-application}
\end{equation}
Here $J_z$ and $J_\pm$ are collective-spin operators for the $n$ two-level emitters, while $a_\mu$ annihilates a photon in mode $\mu$. 
The unknown scalar parameters are $\omega_0,\omega_\mu\in\RR$ and $g_\mu\in\mathbb C$. 
Equivalently, each $g_\mu$ may be parametrized by its real and imaginary parts, which serve as real coefficients of two known Hermitian generators.
For a fixed evolution time $t$, write $U=e^{-itH_{\rm TC}}$. An echo-based OTOC implements the Heisenberg operator $W(t)=U^\dagger WU$, where $W$ is known, and therefore requires the backward evolution $U^\dagger$~\cite{Swingle2016MeasuringScrambling}.
When the forward evolution is an atom--mode block with unknown frequencies and
couplings, this inverse cannot be synthesized from a calibrated Hamiltonian model. 
The protocol has coherent access to both the spin and resolved modes, as in programmable spin--boson platforms~\cite{Sun2025StructuredBath}.

Schur-Weyl duality splits the emitter space into sectors $\cM_j\otimes \cV_{S_j}$, with multiplicity space $\cM_j$ and active spin $S_j=n/2-j$~\cite{BaconChuangHarrow2006Schur}. Equation~\eqref{eq:TC-application} conserves the total-excitation operator $\mathcal N=J_z+\frac n2 I+\sum_\mu a_\mu^\dagger a_\mu$. When we consider states with at most $N$ total excitation, for $n\ge2N$, only spin values $S_j$ with $j\le N$ can occur.
Theorem~\ref{thm:algebraic-inversion} removes $\cM_j$; occupation counting then gives $d_j=\binom{M+N-j+1}{N-j}$. At fixed $M$ and $N$, these active dimensions are independent of $n$.

\begin{prlcorollary}{Tavis--Cummings dynamics}\label{cor:TC-scaling}
Fix $M$ and $N$. For $n\ge2N$, the Tavis--Cummings family restricted to at-most-$N$-excitation subspace admits a clean exact inverse using $\mathcal O_{M,N}(1)$ forward calls. The reversing cost is therefore bounded independently of the number $n$ of emitters~\cite{SupplementalMaterial}.
\end{prlcorollary}

\update{Permutation symmetry therefore changes the dimension-only $\mathcal O(n^{2N})$ benchmark into an $\mathcal O_{M,N}(1)$ forward-query bound at fixed $M$ and $N$.}

\emph{Collective-spin family.---}
As a second example, we consider the representative Ising- and Lipkin--Meshkov--Glick-type Hamiltonian~\cite{Dicke1954,LipkinMeshkovGlick1965,KumariAlhambra2022CollectiveSpin},
\begin{equation}
H_{\rm col}=
\sum_{\mu=x,y,z}h_\mu J_\mu
+\frac1n\sum_{\mu\le\nu}\chi_{\mu\nu}\left(J_\mu J_\nu+J_\nu J_\mu\right)
\label{eq:collective-QEM}
\end{equation}
Here $J_\mu=\frac12\sum_{r=1}^{n}\sigma_r^\mu$ is the collective-spin operator along $\mu\in\{x,y,z\}$, and $\sigma_r^\mu$ is the Pauli operator on spin $r$. The unknown scalar parameters are $h_\mu,\chi_{\mu\nu}\in\RR$.

For a fixed evolution time $t$, write $U=e^{-itH_{\rm col}}$. Echo verification contains the coherent sequence $U^\dagger VU$, where $V$ is a known verification operation. The subsequent verification step may involve postselection~\cite{Weaving2025EVCDR}. For a hardware-native analog block, a
parameter-independent control reversal need not invert all realized
Hamiltonian terms~\cite{KissGrossiRoggero2025QEM}. Our protocol instead
constructs the inverse from repeated uses of the same forward evolution. Because
the inverse is applied unconditionally to the system register, its residual global phase is irrelevant.

Removing the repeated Schur-Weyl factors leaves active dimensions $d_j=n-2j+1$. The automatic construction sums universal inversion costs over these active blocks. Treating the full Hilbert space as an arbitrary unitary would instead give the dimension-only $\mathcal O(4^n)$ benchmark. 

\begin{prlcorollary}{Collective-spin dynamics}\label{cor:collective-scaling}
The collective family in equation~\eqref{eq:collective-QEM} admits a clean exact inverse using $\mathcal O(n^3)$ forward calls~\cite{SupplementalMaterial}. 
\end{prlcorollary}

\emph{Passive multimode links.---}
In the last example, we consider $n$ optical modes with coherent mixing generated by $H_{\rm link}=\sum_{r,s=1}^{n}X_{rs}a_r^\dagger a_s$. Here $a_r$ annihilates a photon in mode $r$, while $X=(X_{rs})\in\mathbb C^{n\times n}$ is the unknown Hermitian single-particle coupling matrix. Thus $X_{rs}$ are scalar coefficients and $H_{\rm link}$ is an operator on the photonic Fock space. Receiver-side meshes can undo a characterized mode transformation~\cite{Lu2024OpticalDescrambler}; here we ask whether the inverse can instead be constructed from repeated uses of the same forward evolution.
We analyze the ideal unitary limit of a stable link that can be traversed coherently using known routing operations, and restrict the input to at most $N$ photons.

\begin{figure}[t]
\centering
\includegraphics[width=0.98\columnwidth]{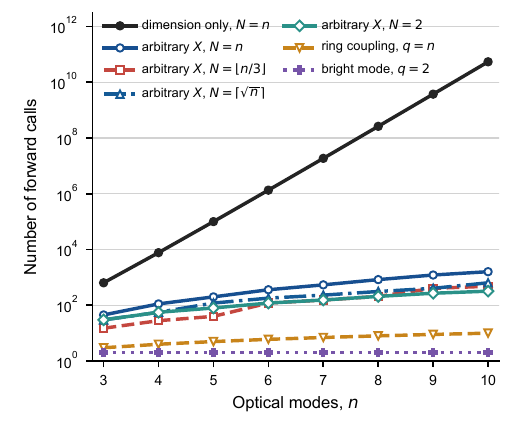}
\caption{\label{fig:link-hierarchy}\textbf{Query hierarchy for a passive multimode link.} 
{The horizontal axis gives the number $n$ of optical modes, the vertical axis is logarithmic, and $N$ denotes the photon-number cutoff. The dimension-only curve treats the full at-most-$n$-photon Hilbert space as an arbitrary unitary. The four arbitrary-$X$ curves show constructive upper bounds for the indicated values of $N$. For $N=2$, translation-invariant couplings around a ring and the rank-one bright-mode family give the exact values $q=n$ and $q=2$, respectively. The widening separation shows how additional structure changes the scaling from exponential to polynomial, and then to linear or constant.}}
\end{figure}

The photon-number operator $\widehat N=\sum_{r=1}^{n}a_r^\dagger a_r$ commutes with $H_{\rm link}$. 
The at-most-$N$-photon space therefore decomposes into $r$-photon subspaces of dimensions $d_r(n)=\binom{n+r-1}{r}$. On each such subspace, the same unknown $n$-mode transformation acts on every photon. Applying the corresponding single-particle inverse to each photon and synchronizing the photon-number components gives an inverse using $q_{\rm pass}=\mathcal O(Nn^2)$ forward calls~\cite{SupplementalMaterial}. For fixed $N$, this is $\mathcal O(n^2)$; it becomes $\mathcal O(n^{5/2})$ for $N=\lceil\sqrt n\rceil$ and $\mathcal O(n^3)$ for $N=n$. For $N=n$, the truncated space has dimension $\binom{2n}{n}$, so treating its evolution as arbitrary gives the dimension-only $\mathcal O(16^n/n)$ benchmark shown in Fig.~\ref{fig:link-hierarchy}.

\begin{prlcorollary}{Passive-link hierarchy}\label{cor:passive-scaling}
For arbitrary Hermitian $X$, the passive link on the at-most-$N$-photon subspace admits a clean exact inverse using $\mathcal O(Nn^2)$ forward calls. For fixed $N$, this gives $\mathcal O(n^{2})$. On the at-most-two-photon subspace, the exact reversing costs are $n$ for Hermitian circulant $X$ and $2$ for the rank-one bright-mode family with $n\ge2$. The latter cost reduces to one when the input is restricted to exactly two excitations~\cite{SupplementalMaterial}.
\end{prlcorollary}

Additional spectral promises reduce these matrix-valued representations.
For translation-invariant couplings around a ring, $X_{rs}$ depends only on $r-s$ modulo $n$. The known Fourier basis then diagonalizes all sectors on the cutoff $\widehat N\le2$, leaving only scalar occupation phases.
Within the fixed-eigenbasis model of Theorem~\ref{thm:spectral-seed}, the optimum is $q=n$; the construction and matching lower bound are given in the Supplemental Material~\cite{SupplementalMaterial}.
For the rank-one bright-mode family $X=x_0I+x_1|b\rangle\!\langle b|$, with $x_0,x_1\in\mathbb R$ and $|b\rangle=n^{-1/2}\sum_r|r\rangle$, only six characters remain. The corresponding fixed-eigenbasis optimum is $q=2$, while one query suffices on the exactly two-excitation sector~\cite{SupplementalMaterial}.

\paragraph*{Discussion.---}
\update{We show that known Hamiltonian structure can reduce the forward-query complexity of deterministic exact inversion.}
For commuting Hamiltonian families, the optimal query number is determined by additive relations among distinct spectral labels, independent of degeneracy. For noncommuting families, Wedderburn reduction removes passive multiplicities. 
\update{Our construction then depends on inverting the inequivalent active blocks and aligning their residual phases.}
Thus, a rapidly growing many-body Hilbert space need not imply a harder inverse problem: much of it may consist only of redundant copies of the same unknown action.

More broadly, inversion and Hamiltonian simulation ask different questions. Simulation cost is typically governed by locality, sparsity, evolution time, and precision. Inversion instead asks whether repeated forward evolutions can be combined coherently to reproduce the inverse on every active sector. The two costs therefore need not track one another: exact spectral relations can make inversion cheap, while perturbations that break those relations may remove the exact advantage. This distinction suggests new directions in approximate and noise-robust inversion. It may also be useful in quantum metrology, where a coherently constructed inverse could implement $U(x)^\dagger V(\theta)U(x)$ for a known control $V(\theta)$ without first estimating the unknown nuisance parameters $x$. Coherent reversibility may therefore be viewed as an independent resource for quantum control and sensing.

\vspace{\baselineskip}
\paragraph*{Acknowledgments.---}
The authors would like to thank  the AQIS 2026 reviewers for thoughtful and insightful comments.
This work was partially supported by the National Key R\&D Program of China (Grant No.~2024YFB4504004), the National Natural Science Foundation of China (Grant. Nos.~2576114,
12447107), the Guangdong Provincial Quantum Science Strategic Initiative (Grant Nos.~GDZX2403008, GDZX2503001), the Guangdong Provincial Key Lab of Integrated Communication, Sensing and Computation for Ubiquitous Internet of Things (Grant No. 2023B1212010007).

\bibliographystyle{apsrev4-2}
\bibliography{ref}

\onecolumngrid
\clearpage
\newtheorem{definition}{Definition}
\newtheorem{proposition}{Proposition}
\newtheorem{lemma}[proposition]{Lemma}
\newtheorem{fact}{Fact}
\newtheorem{theorem}[proposition]{Theorem}
\newtheorem{corollary}[proposition]{Corollary}
\newtheorem{conjecture}{Conjecture}
 
\def\squareforqed{\hbox{\rlap{$\sqcap$}$\sqcup$}}
\def\qed{\ifmmode\squareforqed\else{\unskip\nobreak\hfil
\penalty50\hskip1em\null\nobreak\hfil\squareforqed
\parfillskip=0pt\finalhyphendemerits=0\endgraf}\fi}
\def\endenv{\ifmmode\;\else{\unskip\nobreak\hfil
\penalty50\hskip1em\null\nobreak\hfil\;
\parfillskip=0pt\finalhyphendemerits=0\endgraf}\fi}
\newenvironment{proof}{\noindent \textbf{{Proof~} }}{\hfill $\blacksquare$}

\newcounter{remark}
\newenvironment{remark}[1][]{\refstepcounter{remark}\par\medskip\noindent%
\textbf{Remark~\theremark #1} }{\medskip}

\newcounter{example}
\newenvironment{example}[1][]{\refstepcounter{example}\par\medskip\noindent%
\textbf{Example~\theexample #1} }{\medskip}

\newcommand{\exampleTitle}[1]{\textbf{(#1)}}
\newcommand{\remarkTitle}[1]{\textbf{(#1)}}
\newcommand{\proofComment}[1]{\exampleTitle{#1}}

\newenvironment{beweis}[1]{\noindent\textbf{Proof~{#1}~} }{\qed}

\mathchardef\ordinarycolon\mathcode`\:
\mathcode`\:=\string"8000
\def\vcentcolon{\mathrel{\mathop\ordinarycolon}}
\begingroup \catcode`\:=\active
  \lowercase{\endgroup
  \let :\vcentcolon
  }
\newcommand{\nc}{\providecommand}
\nc{\rnc}{\renewcommand}
\nc{\lbar}[1]{\overline{#1}}
\nc{\bra}[1]{\langle#1|}
\nc{\ket}[1]{|#1\rangle}
\nc{\dketbra}[2]{\vert #1 \rangle \hspace{-.8mm} \rangle \hspace{-.4mm} \langle\hspace{-.8mm}\langle #2 \vert}
\nc{\dbra}[1]{\langle\hspace{-.8mm}\langle #1\vert}
\nc{\dket}[1]{\vert#1\rangle\hspace{-.8mm}\rangle}
\nc{\ketbra}[2]{|#1\rangle\!\langle#2|}
\nc{\braket}[2]{\langle#1|#2\rangle}
\newcommand{\braandket}[3]{\langle #1|#2|#3\rangle}

\nc{\proj}[1]{| #1\rangle\!\langle #1 |}
\nc{\avg}[1]{\langle#1\rangle}
\nc{\rank}{\operatorname{Rank}}
\nc{\smfrac}[2]{\mbox{$\frac{#1}{#2}$}}
\nc{\tr}{\operatorname{Tr}}
\nc{\ox}{\otimes}
\nc{\dg}{\dagger}
\nc{\dn}{\downarrow}
\nc{\cA}{{\cal A}}
\nc{\cB}{{\cal B}}
\nc{\cC}{{\cal C}}
\nc{\cD}{{\cal D}}
\nc{\cE}{{\cal E}}
\nc{\cF}{{\cal F}}
\nc{\cG}{{\cal G}}
\nc{\cH}{{\cal H}}
\nc{\cI}{{\cal I}}
\nc{\cJ}{{\cal J}}
\nc{\cK}{{\cal K}}
\nc{\cL}{{\cal L}}
\nc{\cM}{{\cal M}}
\nc{\cN}{{\cal N}}
\nc{\cO}{{\cal O}}
\nc{\cP}{{\cal P}}
\nc{\cQ}{{\cal Q}}
\nc{\cR}{{\cal R}}
\nc{\cS}{{\cal S}}
\nc{\cT}{{\cal T}}
\nc{\cU}{{\cal U}}
\nc{\cV}{{\cal V}}
\nc{\cX}{{\cal X}}
\nc{\cY}{{\cal Y}}
\nc{\cZ}{{\cal Z}}
\nc{\cW}{{\cal W}}
\nc{\csupp}{{\operatorname{csupp}}}
\nc{\qsupp}{{\operatorname{qsupp}}}
\nc{\var}{{\operatorname{var}}}
\nc{\rar}{\rightarrow}
\nc{\lrar}{\longrightarrow}
\nc{\polylog}{{\operatorname{polylog}}}
\nc{\wt}{{\operatorname{wt}}}
\nc{\av}[1]{{\left\langle {#1} \right\rangle}}
\nc{\supp}{{\operatorname{supp}}}
\nc{\argmin}{{\operatorname{argmin}}}
\nc{\ERLO}{{E_{{ R,LO}}}}
\nc{\ERLOCC}{{E_{{R,\text{PPT}}}}}
\nc{\ERPPT}{{E_{{R,\text{PPT}}}}}
\nc{\ERPPTinf}{{E^{\infty}_{{R,\text{PPT}}}}}
\nc{\ER}{E_{\rm R}}
\nc{\ERLOCCinfty}{{E^{\infty}_{{r,LOCC}}}}
\nc{\Aram}{{\operatorname{\sf A}}}
\nc{\ECPPT}{{E_{{C,\text{PPT}}}}}
\nc{\EDPPT}{{E_{{D,\text{PPT}}}}}
\nc{\Freek}{{\text{PPT$_k$}}}
\nc{\Freesec}{{\text{PPT$_2$}}}

\nc{\NB}{N}
\nc{\LB}{LN}
\nc{\NPT}{{\text{NPT}}}
\nc{\RR}{{{\mathbb R}}}
\nc{\CC}{{{\mathbb C}}}
\nc{\FF}{{{\mathbb F}}}
\nc{\NN}{{{\mathbb N}}}
\nc{\ZZ}{{{\mathbb Z}}}
\nc{\PP}{{{\mathbb P}}}
\nc{\QQ}{{{\mathbb Q}}}
\nc{\UU}{{{\mathbb U}}}
\nc{\EE}{{{\mathbb E}}}
\nc{\id}{{\operatorname{id}}}
\nc{\PPT}{{\text{PPT}}}

\nc{\OLOCC}{{\text{1-LOCC}}}
\nc{\SEP}{{\text{SEP}}}
\nc{\NS}{{\text{NS}}}
\nc{\LOCC}{{\text{LOCC}}}
\nc{\cptp}{\operatorname{CPTP}}
\nc{\hptp}{\operatorname{HPTP}}
\newcommand{\cvartheta}{{\widetilde\vartheta}}
\newcommand{\Op}{\operatorname}

\newtheorem{problem}{Problem}
\newtheorem{note}[problem]{Note}
\newcommand{\Choi}{Choi-Jamio\l{}kowski }
\newcommand{\eps}{\varepsilon}

\nc{\spn}{\operatorname{span}}

\nc{\rep}[1]{\textbf{R}(#1)}

\nc{\dbraket}[2]{\langle\!\langle#1|#2\rangle\!\rangle}

\newcommand{\smcite}[1]{[#1]}
\setcounter{page}{1}
\renewcommand{\thepage}{S\arabic{page}}
\setcounter{section}{0}
\setcounter{equation}{0}
\setcounter{proposition}{0}
\setcounter{definition}{0}
\setcounter{figure}{0}
\setcounter{table}{0}
\renewcommand{\thesection}{S\arabic{section}}
\renewcommand{\thesubsection}{\thesection.\arabic{subsection}}
\renewcommand{\thesubsubsection}{\thesubsection.\arabic{subsubsection}}
\renewcommand{\theequation}{S\arabic{equation}}
\renewcommand{\theproposition}{S\arabic{proposition}}
\renewcommand{\thedefinition}{S\arabic{definition}}
\renewcommand{\thefigure}{S\arabic{figure}}
\renewcommand{\thetable}{S\arabic{table}}
\providecommand{\theHequation}{}
\providecommand{\theHproposition}{}
\providecommand{\theHdefinition}{}
\providecommand{\theHfigure}{}
\providecommand{\theHtable}{}
\renewcommand{\theHequation}{supp.eq.\arabic{equation}}
\renewcommand{\theHproposition}{supp.prop.\arabic{proposition}}
\renewcommand{\theHdefinition}{supp.def.\arabic{definition}}
\renewcommand{\theHfigure}{supp.fig.\arabic{figure}}
\renewcommand{\theHtable}{supp.tab.\arabic{table}}
\allowdisplaybreaks

\begin{center}
{\large\bfseries Supplemental Material for ``Algebraic Speedups for Exact Inversion of Hamiltonian Evolutions''\par}
\vspace{1em}
Jizhe Lai,\textsuperscript{1,2} Mingrui Jing,\textsuperscript{1,2}
Erdong Huang,\textsuperscript{1,2} and Xin Wang\textsuperscript{1,*}\\[0.4em]
\textsuperscript{1}Thrust of Artificial Intelligence, Information Hub,\\
The Hong Kong University of Science and Technology (Guangzhou), Guangzhou 511453, China\\
\textsuperscript{2}QudeLeap Research, Shanghai 200030, China
\end{center}
\vspace{1em}

This Supplemental Material proves the fixed-eigenbasis result and the
Wedderburn reduction used in the paper. It then gives the blockwise inversion
circuits, derives the three application bounds, and shows how additional
block-specific inverses or scalar-phase circuits can improve the automatic
upper bound.

\section{Notations and preliminaries}\label{appendix:Notations and preliminaries}

For a finite-dimensional Hilbert space \(X\), \(\cB(X)\) denotes the algebra of
linear operators on \(X\), \(U(X)\) the unitary group, and \(I_X\) the identity.
All circuits below use the clean convention of
Definition~\ref{def:structured-sequential-protocol}: every work register is
initialized in a fixed state and returned to that state for every input and
every parameter value. Fixed \(x\)-independent changes of the work register are
absorbed into the surrounding fixed gates. All isometries are
implemented by fixed unitary extensions after adding clean ancillas; the
extensions outside the legal subspaces do not affect the circuit action.
The main overloaded symbols are fixed here: \(q\) denotes a query number,
\(K\) denotes the number of distinct spectral values, \(\mathcal A\) denotes the
generated algebra.  To keep this appendix readable, Table~\ref{tab:core-notation}
records only the block-synchronization notation that is easy to confuse;
spectral, model-specific, and traceless-padding symbols are introduced locally
where they are used.

\par\medskip
\begin{center}
\refstepcounter{table}\label{tab:core-notation}
\textbf{TABLE~\thetable.} Core notation for blockwise synchronization.
\par\smallskip
\footnotesize
\begin{tabular}{@{}lll@{}}
Symbol & Variant / definition & Meaning \\
\hline
\(\mathcal A\) & \(\operatorname{alg}^{*}(I,H_1,\ldots,H_p)\) & finite-dimensional \(*\)-algebra generated by the known components \\
\(\Lambda_{\mathrm{blk}}\) &  & active irreducible block labels after Wedderburn reduction \\
\(\cM_\alpha,\cV_\alpha,d_\alpha\) & \(d_\alpha=\dim \cV_\alpha\) & multiplicity space, active block, and block dimension \\
\(U_\alpha,U_{\mathrm{red}}\) & \(U_{\mathrm{red}}=\oplus_\alpha U_\alpha\) & block family and reduced family \\
\(\mathcal P_\alpha,\ell_\alpha,\rho_\alpha\) & \(\mathcal P_\alpha[U_\alpha]=\rho_\alpha U_\alpha^\dagger\) & local inversion primitive \\
\(\mathcal S_\gamma,w_\gamma,\eta_\gamma\) & \(\mathcal S_\gamma[U_{\mathrm{red}}]=\eta_\gamma I\) & scalar phase gadget \\
\(n_{\alpha\gamma},\zeta\) & \(n_{\alpha\gamma}\in\ZZ_{\ge0}\) & gadget copies and synchronized phase \\
\(q_{\mathrm{auto}}\) & \(\sum_\alpha(\ell_\alpha+1)-1\) & automatic completion budget \\
\(s_\alpha\) &  & determinant exponent of a universal block primitive \\
\(D_\alpha\) & \(D_\alpha(x)=\det U_\alpha(x)\) & determinant character \\
\(\tau_\alpha\) & \(D_\alpha(x)=e^{i\tau_\alpha\cdot x},\;(\tau_\alpha)_i=\operatorname{tr}H_{i,\alpha}\) & trace vector \\
\(q_{\mathrm{det}}\) & \(\sum_\alpha d_\alpha s_\alpha-1\) & universal determinant budget \\
\(v^{(\alpha)},\tau_\star\) &  & determinant inventory and synchronized trace vector \\
\hline
\end{tabular}
\end{center}
\medskip

We use the finite-dimensional Wedderburn decomposition in the following
standard form~\smcite{1}.  Let
\[
    \mathcal A:=\operatorname{alg}^{*}(I,H_1,\ldots,H_p)\subseteq\cB(\cH_S).
\]
After a fixed basis change,
\[
    \cH_S\simeq
    \bigoplus_{\alpha\in\Lambda_{\mathrm{blk}}}(\cM_\alpha\otimes \cV_\alpha),
\]
and every \(X\in\mathcal A\) decomposes as
\[
    X\simeq
    \bigoplus_{\alpha\in\Lambda_{\mathrm{blk}}}
    (I_{\cM_\alpha}\otimes X_\alpha).
\]
The spaces \(\cM_\alpha\) are passive multiplicity factors and \(\cV_\alpha\) are
active irreducible factors.  We also use the elementary top-exterior identity:
for \(X\in U(V)\) with \(d=\dim V\), \(X^{\otimes d}\) acts on
\(\bigwedge^d V\) by the scalar \(\det X\).

\section{Abelian scalar-block routing and optimality}\label{sec:supp-proof}

\subsection{A cyclic-shift twirling lemma}\label{subsec:shift-lemma}

Let $\{\ket{j}\}_{j=0}^{K-1}$ be an orthonormal basis of a $K$-dimensional space and
define the cyclic shift
\begin{equation}\label{eq:shift-def-supp}
S:=\sum_{j=0}^{K-1}\ket{j\oplus 1}\!\bra{j},
\end{equation}
where $\oplus$ denotes addition modulo $K$.

\begin{lemma}\label{lem:shift-twirl}
Let $A=\sum_{j=0}^{K-1} a_j\ketbra{j}{j}$ be diagonal in $\{\ket{j}\}$.
Then
\begin{equation}\label{eq:twirl-supp}
\sum_{r=0}^{K-1} S^r A S^{-r} = \operatorname{tr}(A)\,I,
\qquad\text{and hence}\qquad
\sum_{r=1}^{K-1} S^r A S^{-r} = \operatorname{tr}(A)\,I - A.
\end{equation}
\end{lemma}

\begin{proof}
Since $S^r\ket{j}=\ket{j\oplus r}$, we have
$S^r\ketbra{j}{j}S^{-r}=\ketbra{j\oplus r}{j\oplus r}$.
Therefore,
\[
\sum_{r=0}^{K-1} S^r A S^{-r}
= \sum_{r=0}^{K-1}\sum_{j=0}^{K-1} a_j\,\ketbra{j\oplus r}{j\oplus r}.
\]
For each fixed $\ell$, the projector $\ketbra{\ell}{\ell}$ appears once for each $j$, so the coefficient in front of $\ketbra{\ell}{\ell}$ is $\sum_j a_j=\operatorname{tr}(A)$. This proves \eqref{eq:twirl-supp}. Subtracting the $r=0$ term yields the second identity.
\end{proof}

We now set up the fixed scalar-block routing objects used in the proof.  Let the spectral decomposition of \(H\) be
\begin{equation}\label{eq:spectral-decomp-supp}
H=\sum_{k=0}^{K-1}\lambda_k\,\Pi_k,
\qquad \lambda_k\neq \lambda_{k'}~(k\neq k'),
\qquad \Op{rank}[\Pi_k]=m_k.
\end{equation}
Fix an eigenbasis \(\{\ket{k,a}\}\), where \(a\) labels the multiplicity in the
\(\lambda_k\)-eigenspace, so that
\[
    \Pi_k=\sum_{a=1}^{m_k}\ketbra{k,a}{k,a}.
\]
The following fixed maps are used by the scalar-block routing circuit.
Choose an ancilla
$\cH_A\simeq\CC^{\max_k m_k}$ with orthonormal basis
$\{\ket{a-1}_A:1\le a\le\max_k m_k\}$, initialized in $\ket{0}_A$.
\begin{definition}
The multiplicity labeling map is first defined as an isometry on
\(\ket{0}_A\otimes\cH_S\) by
\begin{equation*}
\ket{0}_A\otimes\ket{k,a}_S
\longmapsto
\ket{a-1}_A\otimes\ket{k,1}_S,
\qquad 1\le a\le m_k,
\end{equation*}
where $m_k$ is the multiplicity of the $k$th eigenvalue.  We denote by \(W\)
any fixed unitary extension of this isometry.
\end{definition}
On the legal encoded subspace,
\[
W^\dagger\big(\ket{a-1}_A\otimes\ket{k,1}_S\big)
=\ket{0}_A\otimes\ket{k,a}_S .
\]
We also define the cyclic shift \(S\) on the representative eigenstates by
\begin{equation}\label{eq:S-on-rep-supp}
S\ket{k,1}=\ket{k\oplus 1,1},
\end{equation}
where \(\oplus\) denotes addition modulo \(K\). We extend it to a unitary on the
full \(\cH_S\), for instance by letting it act as identity on the orthogonal
complement of
$\mathrm{span}_{\CC}\{\ket{k,1}\}_{k=0}^{K-1}$).

\subsection{Proof of the universal \texorpdfstring{$K-1$}{K-1} scalar-block routing}
\label{subsec:proof-main}
\begin{proof}
Let $H=\sum_{j=0}^{K-1}\lambda_j \Pi_j$
be the spectral decomposition of \(H\), where the \(\lambda_j\)'s are the \(K\) distinct eigenvalues, and let
\[
\Pi_j=\sum_{a=1}^{m_j}\ket{j,a}\!\bra{j,a}
\]
where \(\{\ket{j,a}:1\le a\le m_j\}\) is an orthonormal eigenbasis and
$m_j$ is the multiplicity of $\lambda_j$. The evolution satisfies
\[
U(t)=e^{iHt}
\qquad\text{and}\qquad
U(t)\ket{j,a}=e^{i\lambda_j t}\ket{j,a}.
\]
For any input state $\ket{\psi}_S$ of the system $\cH_S$, we can expand it in
this basis as
\[
\ket{\psi}_S=\sum_{j=0}^{K-1}\sum_{a=1}^{m_j}\alpha_{j,a}\ket{j,a}_S.
\]
By the definition of the labeling unitary \(W\), we have
\begin{equation}\label{eq:proof-W-action}
W(\ket{0}_A\otimes\ket{\psi}_S)
=
\sum_{j=0}^{K-1}\sum_{a=1}^{m_j}
\alpha_{j,a}\,\ket{a-1}_A\otimes \ket{j,1}_S.
\end{equation}
Hence, after the encoding, the system register is supported on the representative subspace
\[
\cH_R:=\operatorname{span}_{\CC}\{\ket{j,1}:0\le j\le K-1\}.
\]
Let \(S\) denote the cyclic shift on $\cH_R$, defined in Eq.~\eqref{eq:S-on-rep-supp}.
Consider the \((K-1)\)-query unitary
\begin{equation}\label{eq:proof-coreUt}
\mathcal R_t:=\prod_{r=1}^{K-1}\bigl(S^r U(t) S^{-r}\bigr).
\end{equation}
Products are written in algebraic right-to-left order. In the circuit
diagram the same slot is drawn left-to-right as \(S^{-r}\), then \(U(t)\), then
\(S^r\); the two notations describe the same conjugated query.
Each factor contains exactly one call to \(U(t)\), so \(\mathcal R_t\) uses \(K-1\) queries in total.

We now compute the action of \(\mathcal R_t\) on the basis vectors \(\ket{j,1}\) of \(\cH_R\). For every \(r\in\{1,\dots,K-1\}\),
\[
S^{-r}\ket{j,1}=\ket{j\ominus r,1},
\]
where \(\ominus\) denotes subtraction modulo \(K\). Therefore,
\[
(S^r U(t) S^{-r})\ket{j,1}
=
S^r U(t) \ket{j\ominus r,1}
=
S^r\bigl(e^{i\lambda_{j\ominus r}t}\ket{j\ominus r,1}\bigr)
=
e^{i\lambda_{j\ominus r}t}\ket{j,1}.
\]
Multiplying these phases over \(r=1,\dots,K-1\), we obtain
\begin{equation}\label{eq:proof-phase-product}
\mathcal R_t\ket{j,1}
=
\exp\!\Bigl(it\sum_{r=1}^{K-1}\lambda_{j\ominus r}\Bigr)\ket{j,1}.
\end{equation}
Let
$\beta:=\sum_{\ell=0}^{K-1}\lambda_\ell$. Since
$\{j\ominus r : 1\le r\le K-1\}
=\{0,1,\dots,K-1\}\setminus\{j\}$, it follows that
\[
\sum_{r=1}^{K-1}\lambda_{j\ominus r}
=
\beta-\lambda_j.
\]
Substituting this into \eqref{eq:proof-phase-product} yields
\begin{equation}\label{eq:proof-rep-action}
\mathcal R_t\ket{j,1}
=
e^{it(\beta-\lambda_j)}\ket{j,1}
=
e^{it\beta}e^{-i\lambda_j t}\ket{j,1}.
\end{equation}
Applying \(I_A\otimes \mathcal R_t\) to \eqref{eq:proof-W-action} and using \eqref{eq:proof-rep-action}, we get
\[
(I_A\otimes \mathcal R_t)\,W(\ket{0}_A\otimes\ket{\psi}_S)
=
e^{it\beta}
\sum_{j=0}^{K-1}\sum_{a=1}^{m_j}
\alpha_{j,a}\,\ket{a-1}_A\otimes e^{-i\lambda_j t}\ket{j,1}_S.
\]
Finally, applying \(W^\dagger\) and using the decoding relation
\[
W^\dagger(\ket{a-1}_A\otimes \ket{j,1}_S)
=
\ket{0}_A\otimes \ket{j,a}_S,
\]
we obtain
\[
W^\dagger(I_A\otimes \mathcal R_t)W(\ket{0}_A\otimes\ket{\psi}_S)
=
\ket{0}_A\otimes
\Bigl(
e^{it\beta}
\sum_{j=0}^{K-1}\sum_{a=1}^{m_j}
\alpha_{j,a}e^{-i\lambda_j t}\ket{j,a}_S
\Bigr).
\]
Since \(U(t)^\dagger\ket{j,a}=e^{-i\lambda_j t}\ket{j,a}\), the system output is exactly
\[
e^{it\beta}U(t)^\dagger\ket{\psi}_S.
\]
Thus, for every input state \(\ket{\psi}_S\) and every \(t\),
\[
W^\dagger(I_A\otimes \mathcal R_t)W(\ket{0}_A\otimes\ket{\psi}_S)
=
\ket{0}_A\otimes e^{it\beta}U(t)^\dagger\ket{\psi}_S.
\]
Therefore the protocol deterministically implements
\[
\mathcal C(U(t))=e^{it\beta}U(t)^\dagger
\]
using exactly \(K-1\) queries to \(U(t)\).
\end{proof}

\subsection{Proof of the spectral sumset optimum and exact \texorpdfstring{$q$}{q}-slot construction}
\label{appendix:q-barycentric}

We now restate and prove the fixed-eigenbasis result from the Letter. The proof
first relates the matrix and sumset formulations, then proves the lower bound
for an arbitrary exact \(q\)-query circuit, and finally constructs a circuit
whenever the sumset condition holds.

\subsubsection{Equivalent formulations}

For $q\in\ZZ_{\ge0}$, define the $q$-fold sumset
\begin{equation}
    \Sigma_q(\Lambda)
    :=\left\{\mu_1+\cdots+\mu_q:\mu_r\in\Lambda\right\},
    \qquad
    \Sigma_0(\Lambda):=\{0\},
    \label{eq:supp-sumset-definition}
\end{equation}
where the $q=0$ sum is the empty sum.

\begin{theorem}[Fixed-eigenbasis exact query number]
\label{thm:supp-spectral-seed}
Let
\begin{equation}
    H=\sum_{j=0}^{K-1}\lambda_j\Pi_j,
    \qquad
    U(t)=e^{iHt},
    \label{eq:supp-fixed-basis-family}
\end{equation}
where the nonzero projectors \(\Pi_j\) are fixed, pairwise orthogonal, and sum
to the identity, and the real eigenvalues in
\(\Lambda=\{\lambda_0,\ldots,\lambda_{K-1}\}\) are distinct. The minimum
number of forward calls needed to implement \(U(t)^\dagger\), up to a global
phase, for every unknown \(t\in\RR\) is
\begin{equation}
    \kappa(\Lambda)
    =\min\bigl\{q\in\ZZ_{\ge0}:\exists c\in\RR,\;
    c-\Lambda\subseteq\Sigma_q(\Lambda)\bigr\}.
    \label{eq:supp-fixed-basis-optimum}
\end{equation}
In particular, \(\kappa(\Lambda)\le K-1\), independently of the eigenspace
multiplicities.
\end{theorem}

We first prove that the matrix synchronization condition
\begin{equation}
\begin{gathered}
    \exists\,c\in\RR,\quad
    \exists\,N\in\ZZ_{\ge0}^{K\times K}:\\
    N\mathbf{1}_K=q\mathbf{1}_K,
    \qquad
    N\bm{\lambda}=c\mathbf{1}_K-\bm{\lambda}
\end{gathered}
\label{eq:q-min-matrix}
\end{equation}
and the sumset condition
\begin{equation}
    \exists\,c\in\RR:
    \qquad
    c-\Lambda\subseteq\Sigma_q(\Lambda)
\label{eq:q-min-sumset}
\end{equation}
are equivalent.

Suppose that
\[
N\mathbf{1}_K=q\mathbf{1}_K,
\qquad
N\bm{\lambda}=c\mathbf{1}_K-\bm{\lambda},
\]
for some $N=(n_{j\ell})\in\mathbb{Z}_{\ge 0}^{K\times K}$. Then, for each
$j$,
\[
c-\lambda_j=\sum_{\ell=0}^{K-1} n_{j\ell}\lambda_\ell \in \Sigma_q(\Lambda),
\]
because the $j$th row sum of $N$ is equal to $q$. Hence
\[
c-\Lambda\subseteq \Sigma_q(\Lambda).
\]

Conversely, suppose that
\[
c-\Lambda\subseteq \Sigma_q(\Lambda).
\]
Then, for each $j$, one may choose a decomposition
\[
c-\lambda_j=\mu_{j,1}+\cdots+\mu_{j,q},
\qquad
\mu_{j,r}\in\Lambda.
\]
Let $n_{j\ell}$ denote the number of occurrences of $\lambda_\ell$ among
$\mu_{j,1},\ldots,\mu_{j,q}$. This yields a nonnegative integer matrix
$N=(n_{j\ell})$ satisfying
\[
N\mathbf{1}_K=q\mathbf{1}_K,
\qquad
N\bm{\lambda}=c\mathbf{1}_K-\bm{\lambda}.
\]
Therefore Eqs.~\eqref{eq:q-min-matrix} and \eqref{eq:q-min-sumset} are
equivalent.

Finally, writing $c=(q+1)m$ gives
\[
N(\bm{\lambda}-m\mathbf{1}_K)=-(\bm{\lambda}-m\mathbf{1}_K),
\]
which is the centered matrix form of the spectral sumset condition.

\subsubsection{Necessary condition}

Consider an arbitrary clean deterministic exact $q$-slot inversion circuit with
ancilla $A$,
\begin{equation}
M_t
=
(\bra{0}_A\otimes I_S)
G_q (I_A\otimes U(t)) G_{q-1}\cdots
(I_A\otimes U(t)) G_0
(\ket{0}_A\otimes I_S),
\label{eq:general-q-slot-circuit}
\end{equation}
where the unitaries $G_0,\ldots,G_q$ are independent of $t$.

Using the spectral decomposition
\[
U(t)=\sum_{j=0}^{K-1} e^{i\lambda_j t}\Pi_j,
\]
each query in Eq.~\eqref{eq:general-q-slot-circuit} contributes one of the
frequencies $\lambda_0,\ldots,\lambda_{K-1}$. Expanding all $q$ queries
therefore gives
\begin{equation}
M_t
=
\sum_{j_1,\ldots,j_q=0}^{K-1}
e^{i(\lambda_{j_1}+\cdots+\lambda_{j_q})t}
A_{j_1,\ldots,j_q}
=
\sum_{\omega\in \Sigma_q(\Lambda)} e^{i\omega t} A_\omega,
\label{eq:q-slot-frequency-expansion}
\end{equation}
for suitable operators $A_{j_1,\ldots,j_q}$ and $A_\omega$.

Now assume that the circuit implements the inverse evolution up to a global
phase, namely,
\[
M_t=g(t)\,U(t)^\dagger,
\qquad
|g(t)|=1,
\]
for all $t\in\mathbb{R}$. Multiplying on the right by $\Pi_j$ yields
\[
M_t\Pi_j=g(t)e^{-i\lambda_j t}\Pi_j.
\]
By Eq.~\eqref{eq:q-slot-frequency-expansion}, each operator-valued function
$M_t\Pi_j$ is a finite linear combination of exponentials whose frequencies
belong to $\Sigma_q(\Lambda)$.  Pick any index $j$ for which $\Pi_j\ne 0$
(every $j$ qualifies since $\Pi_j$ is a non-zero spectral projector).  The
identity $M_t\Pi_j=g(t)e^{-i\lambda_j t}\Pi_j$ then forces
\[
g(t)e^{-i\lambda_j t}\;\in\;\operatorname{span}_{\CC}
\bigl\{e^{i\omega t}:\omega\in\Sigma_q(\Lambda)\bigr\},
\]
since the right-hand operator has the same scalar Fourier content as the
left.  Equivalently, taking a matrix element inside the nonzero
eigenspace \(\Pi_j\) shows directly that the scalar function
\(g(t)e^{-i\lambda_j t}\) has Fourier support contained in
\(\Sigma_q(\Lambda)\). In particular $g$ itself is a \emph{finite} trigonometric
polynomial---a fact we use below---and its frequency support is contained in
$\lambda_j+\Sigma_q(\Lambda)$.  Varying $j$, the support of $g$ must lie in
\[
\bigcap_{j=0}^{K-1}\bigl(\lambda_j+\Sigma_q(\Lambda)\bigr).
\]
The non-vanishing $|g(t)|=1$ rules out the degenerate case $g\equiv 0$.

Write
\[
g(t)=\sum_{r=1}^{R}\alpha_r e^{i\omega_r t},
\qquad
\omega_1<\omega_2<\cdots<\omega_R.
\]
If $R\ge 2$, then
\[
|g(t)|^2
=
\sum_{r,s=1}^{R}
\alpha_r\overline{\alpha_s}\,e^{i(\omega_r-\omega_s)t}.
\]
The largest frequency appearing on the right-hand side is
$\omega_R-\omega_1>0$, and it arises only from the pair $(r,s)=(R,1)$.
Since $|g(t)|^2=1$ contains no nonzero frequency, its coefficient must
vanish:
\[
\alpha_R\overline{\alpha_1}=0,
\]
which is impossible if both coefficients are nonzero. Hence $R=1$, and
there exist real numbers $c$ and $\theta_0$ such that
\[
g(t)=e^{ict+i\theta_0}.
\]
Absorbing the constant phase $e^{i\theta_0}$ into the fixed outer unitaries,
we may take
\[
g(t)=e^{ict}.
\]
Consequently,
\[
c-\lambda_j\in \Sigma_q(\Lambda),
\qquad
j=0,\ldots,K-1.
\]
This proves the necessity of Eq.~\eqref{eq:q-min-sumset}.

\subsubsection{Sufficiency and exact \texorpdfstring{$q$}{q}-slot construction}

Assume now that Eq.~\eqref{eq:q-min-matrix} holds for some integer $q$,
real number $c$, and nonnegative integer matrix
$N=(n_{j\ell})\in\mathbb{Z}_{\ge 0}^{K\times K}$.
For \(q=0\), the condition holds if and only if all eigenvalues coincide.
Because the labels in \(\Lambda\) are distinct, this is exactly the case
\(K=1\). We therefore describe the construction for \(q\ge1\).

Choose an orthonormal eigenbasis of $H$,
\[
H|j,a\rangle=\lambda_j |j,a\rangle,
\qquad
a=1,\ldots,m_j,
\]
where \(m_j=\mathrm{rank}(\Pi_j)\). For each distinct eigenvalue $\lambda_j$,
fix one representative eigenvector $|j,1\rangle$.

For every row $j$ of $N$, choose an ordered $q$-tuple
\[
(\mu_{j,1},\ldots,\mu_{j,q})
\]
such that each $\lambda_\ell$ appears exactly $n_{j\ell}$ times and
\[
\mu_{j,1}+\cdots+\mu_{j,q}=c-\lambda_j.
\]
Let $\ell_r(j)$ be the index determined by
\[
\mu_{j,r}=\lambda_{\ell_r(j)}.
\]

Introduce an ancilla composed of a spectral-label register $L$, a
multiplicity register $M$, and a reusable query register $Q$ isomorphic to
the system register. Let $|\varnothing\rangle_Q$ be a fixed reference state
of $Q$. Since arbitrarily large ancilla are allowed, there exists an
encoding isometry
\[
W:\ |0\rangle_A\otimes |j,a\rangle_S
\longmapsto
|j\rangle_L\otimes |a\rangle_M\otimes |\varnothing\rangle_Q .
\]
We regard $W$ as extended to a fixed unitary on a larger space.

For each slot $r=1,\ldots,q$, define a controlled preparation map on the
encoded subspace by
\[
R_r:\ |j\rangle_L\otimes |\varnothing\rangle_Q
\longmapsto
|j\rangle_L\otimes |\ell_r(j),1\rangle_Q,
\]
and let it act as the identity on the multiplicity register $M$. Because the
states
\[
|j\rangle_L\otimes |\ell_r(j),1\rangle_Q
\]
are orthonormal as $j$ varies, each $R_r$ extends to a full unitary.

The exact $q$-slot inversion circuit is then
\begin{equation}
\mathcal{C}_t
=
W^\dagger
\Bigl(
R_q^\dagger U(t)^{(Q)} R_q
\Bigr)
\cdots
\Bigl(
R_1^\dagger U(t)^{(Q)} R_1
\Bigr)
W,
\label{eq:q-slot-constructive-circuit}
\end{equation}
where $U(t)^{(Q)}$ denotes $U(t)$ applied to the query register
$Q$ and the identity on all other registers.
If the circuit contains several system-isomorphic registers, one forward call
means that fixed swaps move \(Q\simeq\cH_S\) to the input port, apply \(U(t)\)
once, and move \(Q\) back.

Its action on an encoded basis state is immediate:
\[
R_r^\dagger U(t)^{(Q)} R_r
\bigl(
|j\rangle_L|a\rangle_M|\varnothing\rangle_Q
\bigr)
=
e^{i\mu_{j,r}t}
|j\rangle_L|a\rangle_M|\varnothing\rangle_Q.
\]
Therefore,
\[
\Bigl(
R_q^\dagger U(t)^{(Q)} R_q
\Bigr)
\cdots
\Bigl(
R_1^\dagger U(t)^{(Q)} R_1
\Bigr)
|j\rangle_L|a\rangle_M|\varnothing\rangle_Q
=
e^{i(\mu_{j,1}+\cdots+\mu_{j,q})t}
|j\rangle_L|a\rangle_M|\varnothing\rangle_Q.
\]
Using $\mu_{j,1}+\cdots+\mu_{j,q}=c-\lambda_j$, we obtain
\[
\mathcal{C}_t |j,a\rangle_S
=
e^{i(c-\lambda_j)t}|j,a\rangle_S
=
e^{ict}U(t)^\dagger |j,a\rangle_S.
\]
Hence
\[
\mathcal{C}_t=e^{ict}U(t)^\dagger
\]
for all $t\in\mathbb{R}$, and the protocol uses exactly $q$ queries.

This construction also makes clear why the multiplicities do not affect the
minimum query number. The multiplicity label $a$ is stored in the register
$M$ and never undergoes \(U(t)\). All routing choices depend only on
the spectral label $j$, and therefore only the set of distinct eigenvalues
enters the optimization.

\begin{corollary}[Multiparameter fixed-basis sumset optimum]
\label{cor:multiparameter-sumset}
Let
\begin{equation}
  U(x)=\sum_{j=0}^{K-1}e^{i\lambda_j\cdot x}\Pi_j,
  \qquad x\in\RR^p,
  \qquad \Lambda=\{\lambda_0,\ldots,\lambda_{K-1}\}\subset\RR^p,
  \label{eq:multiparameter-fixed-basis}
\end{equation}
where the nonzero projectors $\Pi_j$ are fixed, pairwise orthogonal, and sum to
the identity, and the character vectors $\lambda_j$ are distinct. The minimum
clean deterministic query number is
\begin{equation}
  \kappa(\Lambda)
  =\min\bigl\{q:\exists c\in\RR^p,\;
  c-\Lambda\subseteq\Sigma_q(\Lambda)\bigr\},
  \label{eq:multiparameter-sumset-optimum}
\end{equation}
where $\Sigma_q(\Lambda)$ is the set of sums of $q$ character vectors.
\end{corollary}

\begin{proof}
The constructive circuit above is unchanged: if
$c-\lambda_j=\mu_{j,1}+\cdots+\mu_{j,q}$, routing slot $r$ through a
representative state with character $\mu_{j,r}$ accumulates
$e^{i(c-\lambda_j)\cdot x}$ on sector $j$.

For necessity, expanding a $q$-query circuit gives an operator-valued finite
exponential polynomial with frequencies in $\Sigma_q(\Lambda)$. Exact
inversion implies
\[
  M(x)\Pi_j=g(x)e^{-i\lambda_j\cdot x}\Pi_j,
  \qquad |g(x)|=1,
\]
so $g$ is itself a finite exponential polynomial. If its support contained
more than one vector, choose $v\in\RR^p$ that separates all frequencies in
that finite support. Then $g(tv)$ would be a one-variable finite Fourier
series with at least two distinct frequencies and unit modulus, contradicting
the one-variable largest-frequency argument above. Hence
$g(x)=e^{i(c\cdot x+\theta_0)}$ for some $c\in\RR^p$. Removing the fixed phase
$e^{i\theta_0}$ gives
$c-\lambda_j\in\Sigma_q(\Lambda)$ for every $j$, proving
Eq.~\eqref{eq:multiparameter-sumset-optimum}.
\end{proof}

\section{Wedderburn reduction and blockwise compiler}
\label{appendix:structured-blockwise}

\subsection{Circuit convention and Wedderburn decomposition}

Consider the structured family
\begin{equation}
    H(x)=\sum_{i=1}^{p}x_iH_i,
    \qquad x\in\RR^p,
    \qquad U(x)=e^{iH(x)},
    \label{eq:structured-family}
\end{equation}
where \(H_1,\ldots,H_p\) are fixed Hermitian operators on \(\cH_S\).  Let
\begin{equation}
    \mathcal A:=\operatorname{alg}^{*}(I,H_1,\ldots,H_p)
    \subseteq \cB(\cH_S).
    \label{eq:structured-algebra}
\end{equation}
By the finite-dimensional Wedderburn theorem, after a fixed basis change,
\begin{equation}
    \cH_S\cong
    \bigoplus_{\alpha\in\Lambda_{\mathrm{blk}}}
    \bigl(\cM_\alpha\otimes \cV_\alpha\bigr),
    \label{eq:structured-wedderburn}
\end{equation}
where \(\cM_\alpha\) is a multiplicity space and \(\cV_\alpha\) is the active
irreducible block.  In this basis,
\begin{equation}
    H_i\cong
    \bigoplus_{\alpha\in\Lambda_{\mathrm{blk}}}
    \bigl(I_{\cM_\alpha}\otimes H_{i,\alpha}\bigr),
    \qquad
    U(x)\cong
    \bigoplus_{\alpha\in\Lambda_{\mathrm{blk}}}
    \bigl(I_{\cM_\alpha}\otimes U_\alpha(x)\bigr),
    \label{eq:structured-block-family}
\end{equation}
where
\begin{equation}
    H_\alpha(x):=\sum_{i=1}^{p}x_iH_{i,\alpha},
    \qquad
    U_\alpha(x):=e^{iH_\alpha(x)}.
    \label{eq:structured-block-unitaries}
\end{equation}
The reduced system keeps one active copy of every block:
\begin{equation}
    \cH_{\mathrm{red}}
    :=\bigoplus_{\alpha\in\Lambda_{\mathrm{blk}}}\cV_\alpha,
    \qquad
    U_{\mathrm{red}}(x)
    :=\bigoplus_{\alpha\in\Lambda_{\mathrm{blk}}}U_\alpha(x).
    \label{eq:structured-reduced-system}
\end{equation}
We write
\begin{equation}
    d_\alpha:=\dim \cV_\alpha.
    \label{eq:structured-block-dimensions}
\end{equation}

All isometries appearing below are implemented as fixed unitaries after
adding clean ancillary registers of sufficient dimension.  We use the same
symbol for the isometry and for any fixed unitary extension, since the
circuits only visit the legal subspaces on which the isometry is defined.

\begin{definition}
\label{def:structured-sequential-protocol}
Let \(X\) be a finite-dimensional Hilbert space and let
\(V_X(x)\in U(X)\) be a unitary family. An \(n\)-query circuit for
\(V_X\) uses a finite-dimensional ancilla \(A\) initialized in a fixed state
\(\ket{0}_A\), fixed \(x\)-independent unitaries, and \(n\) ordinary forward
calls to \(V_X(x)\). Its full unitary circuit is
\begin{equation}
\begin{aligned}
    \widetilde{\mathcal C}_x
    :=G_n(I_A\otimes V_X(x))G_{n-1}\cdots
    (I_A\otimes V_X(x))G_0,
\end{aligned}
    \label{eq:structured-protocol-output}
\end{equation}
for fixed unitaries \(G_0,\ldots,G_n\). The circuit is clean if there is an
operator \(\mathcal C[V_X](x)\) on \(X\) such that
\begin{equation}
    \widetilde{\mathcal C}_x
    \bigl(\ket{0}_A\otimes\ket{\psi}_X\bigr)
    =\ket{0}_A\otimes\mathcal C[V_X](x)\ket{\psi}_X
    \qquad
    \text{for all }x\text{ and }\ket{\psi}_X.
    \label{eq:structured-clean-output}
\end{equation}
The case \(n=0\) means a fixed unitary circuit with no forward call. The
circuit exactly inverts \(V_X\) if there is a phase function
\(\omega(x)\in U(1)\) such that
\begin{equation}
    \mathcal C[V_X](x)=\omega(x)V_X(x)^\dagger
    \qquad\text{for all }x.
    \label{eq:structured-exact-inversion}
\end{equation}
Thus the same global phase \(\omega(x)\) applies to every input state and every
coherent sector. Fixed \(x\)-independent changes of the work register are
absorbed into the surrounding fixed gates. Only the number of forward calls is
counted. Each call is the ordinary operation \(V_X(x)\) on the designated input
port; only the surrounding fixed gates may be coherently controlled.
If the algorithm uses
several registers isomorphic to \(X\), one query means that a fixed swap
network moves the chosen register to the input port, applies \(V_X\), and
moves it back.
\end{definition}

\begin{lemma}
\label{lem:universal-inverter-character}
Let $X\simeq\CC^d$, and let $\mathcal Q$ be a clean deterministic protocol
with exactly $q$ forward calls to an unknown unitary \(G\in U(d)\).
Suppose that for every \(G\in U(d)\),
\begin{equation}
    \mathcal Q[G]=\omega(G)G^\dagger
\end{equation}
for some phase \(\omega(G)\in U(1)\). Then there exists an integer $s\ge 1$
such that
\begin{equation}
    q+1=ds ,
\end{equation}
and, after multiplying the protocol by a fixed \(G\)-independent phase,
\begin{equation}
    \mathcal Q[G]=(\det G)^s G^\dagger
    \qquad
    \forall\,G\in U(d).
\end{equation}
\end{lemma}

\begin{proof}
Let $Z=(z_{ij})$ be a formal $d\times d$ complex matrix. Since the protocol
has exactly $q$ forward calls and is clean, every matrix entry of
$\mathcal Q[Z]$ is a homogeneous polynomial of degree $q$ in the entries of
$Z$. Define
\begin{equation}
    f(Z):=\frac{1}{d}\operatorname{tr}\bigl(\mathcal Q[Z]Z\bigr).
\end{equation}
Then $f$ is homogeneous of degree $q+1$.

For every unitary \(G\), the exact inversion assumption gives
\begin{equation}
    \mathcal Q[G]G=\omega(G)I_d .
\end{equation}
Therefore \(f(G)=\omega(G)\), and hence
\begin{equation}
    \mathcal Q[G]G=f(G)I_d
    \qquad
    \forall\,G\in U(d).
\end{equation}
We use the standard fact that \(U(d)\) is Zariski dense in \(M_d(\CC)\).
Indeed, suppose that a polynomial \(p\) vanishes on \(U(d)\).  For any unitary
\(G\) and Hermitian \(H\), the function \(z\mapsto p(Ge^{izH})\) is entire and
vanishes for all real \(z\), since \(Ge^{izH}\) is unitary for real \(z\).
Hence it vanishes for all complex \(z\), in particular at \(z=-i\), so
\(p(Ge^H)=0\).  By polar decomposition, every element of \(GL_d(\CC)\) has the
form \(Ge^H\) with \(G\) unitary and \(H\) Hermitian.  Thus \(p\) vanishes on
\(GL_d(\CC)\).  Since \(GL_d(\CC)\) is Zariski open and dense in \(M_d(\CC)\),
\(p\) is identically zero.  Therefore the polynomial identity above extends
from \(U(d)\) to all complex matrices.  Applied to each matrix
entry of \(\mathcal Q[Z]Z-f(Z)I_d\), this yields the polynomial identity
\begin{equation}
    \mathcal Q[Z]Z=f(Z)I_d .
    \label{eq:clean-inverter-polynomial-identity}
\end{equation}

Taking determinants in Eq.~\eqref{eq:clean-inverter-polynomial-identity}
gives
\begin{equation}
    \det\bigl(\mathcal Q[Z]\bigr)\det Z=f(Z)^d .
\end{equation}
The ring $\CC[z_{ij}]$ is a unique factorization domain, and $\det Z$ is
irreducible, hence prime. Thus $\det Z\mid f(Z)^d$, so $\det Z\mid f(Z)$.
Since \(f(G)=\omega(G)\) on \(U(d)\), \(f\) is not the zero polynomial. Hence
there is a homogeneous polynomial $h$ such that
\begin{equation}
    f(Z)=(\det Z)h(Z).
\end{equation}
In particular $q+1\ge d$, and $h$ has degree
\begin{equation}
    N:=q+1-d .
\end{equation}

For invertible $Z$, Eq.~\eqref{eq:clean-inverter-polynomial-identity} gives
\begin{equation}
    \mathcal Q[Z]=f(Z)Z^{-1}=h(Z)\operatorname{adj}(Z).
\end{equation}
Both sides are polynomial in $Z$, so the identity extends from
$GL_d(\CC)$ to all complex matrices:
\begin{equation}
    \mathcal Q[Z]=h(Z)\operatorname{adj}(Z).
    \label{eq:clean-inverter-adjugate-form}
\end{equation}

For unitary \(G\), one has
\begin{equation}
    \operatorname{adj}(G)=(\det G)G^\dagger .
\end{equation}
Combining this with Eq.~\eqref{eq:clean-inverter-adjugate-form} gives
\begin{equation}
    \mathcal Q[G]=h(G)(\det G)G^\dagger .
\end{equation}
Comparing with \(\mathcal Q[G]=\omega(G)G^\dagger\), we obtain
\begin{equation}
    \omega(G)=h(G)\det G .
\end{equation}
Since \(|\omega(G)|=1\) and \(|\det G|=1\), it follows that
\begin{equation}
    |h(G)|=1
    \qquad
    \forall\,G\in U(d).
    \label{eq:h-unit-modulus}
\end{equation}

We now show that \(h\) is a determinant power. Let
\begin{equation}
    h^\#(Z):=\overline{h(\overline Z)}
\end{equation}
be the coefficientwise conjugate polynomial. The polynomial \(h^\#\) is
homogeneous of degree \(N\). For unitary \(G\),
using \(\overline G=G^{-T}\) and \(G^{-T}=(\det G)^{-1}\operatorname{adj}(G)^T\),
we get
\begin{equation}
    \overline{h(G)}
    =
    h^\#(\overline G)
    =
    h^\#(G^{-T})
    =
    (\det G)^{-N}h^\#(\operatorname{adj}(G)^T).
\end{equation}
Together with Eq.~\eqref{eq:h-unit-modulus}, this gives
\begin{equation}
    h(G)h^\#(\operatorname{adj}(G)^T)
    =
    (\det G)^N
    \qquad
    \forall\,G\in U(d).
\end{equation}
By Zariski density, this is a polynomial identity:
\begin{equation}
    h(Z)h^\#(\operatorname{adj}(Z)^T)
    =
    (\det Z)^N .
    \label{eq:h-determinant-power}
\end{equation}

Since $\CC[z_{ij}]$ is a unique factorization domain and $\det Z$ is
irreducible, Eq.~\eqref{eq:h-determinant-power} implies that every
irreducible factor of $h$ is associated to $\det Z$. Hence
\begin{equation}
    h(Z)=c(\det Z)^r
\end{equation}
for some $c\in\CC^\times$ and some integer $r\ge 0$. From
Eq.~\eqref{eq:h-unit-modulus}, we have $|c|=1$. Since $h$ is homogeneous of
degree $N$, we also have
\begin{equation}
    N=rd .
\end{equation}
Substituting into Eq.~\eqref{eq:clean-inverter-adjugate-form}, for every
unitary \(G\),
\begin{equation}
    \mathcal Q[G]
    =
    c(\det G)^r\operatorname{adj}(G)
    =
    c(\det G)^{r+1}G^\dagger .
\end{equation}
Set $s:=r+1$. Then $s\ge 1$ and
\begin{equation}
    q+1=d+N=d(r+1)=ds .
\end{equation}
Finally, multiplying the whole protocol by the fixed phase $c^{-1}$ gives
\begin{equation}
    \mathcal Q[G]=(\det G)^sG^\dagger
    \qquad
    \forall\,G\in U(d).
\end{equation}
\end{proof}

The lemma is a statement about clean universal inverters.  Protocols that
leave an input-dependent catalyst or only implement the inverse channel do not
provide the operator-level phase tracked in the blockwise compiler.
The lemma applies to clean universal inverters on \(U(d)\).
Family-specific local primitives are supplied directly to the blockwise
synchronization theorem below.

The following two fixed wirings are the only places where multiplicities
enter.  Let
\begin{equation}
    \cH_{\mathrm{mult}}
    :=\bigoplus_{\alpha\in\Lambda_{\mathrm{blk}}}\cM_\alpha,
    \label{eq:structured-mult-space}
\end{equation}
and let
\(\iota^S_\alpha:\cM_\alpha\otimes \cV_\alpha\hookrightarrow \cH_S\),
\(\iota^R_\alpha:\cV_\alpha\hookrightarrow\cH_{\mathrm{red}}\), and
\(\iota^M_\alpha:\cM_\alpha\hookrightarrow\cH_{\mathrm{mult}}\) be the canonical
inclusions.  Define the \emph{split wiring}
\begin{equation}
    \Xi:\cH_S\longrightarrow \cH_{\mathrm{mult}}\otimes\cH_{\mathrm{red}}
    \label{eq:structured-Xi-map}
\end{equation}
by
\begin{equation}
    \Xi\,\iota^S_\alpha(m\otimes v)
    :=\iota^M_\alpha(m)\otimes\iota^R_\alpha(v).
    \label{eq:structured-Xi-definition}
\end{equation}
This gate pulls the inactive copy label \(m\) out of the physical system and
leaves the active vector \(v\) in the reduced register.  Conversely, choose
one reference copy in each multiplicity space,
\begin{equation}
    \ket{\mathrm{ref}_\alpha}\in \cM_\alpha,
    \qquad \|\ket{\mathrm{ref}_\alpha}\|=1,
    \label{eq:structured-reference-vectors}
\end{equation}
and define the \emph{load wiring}
\begin{equation}
    \Upsilon:\cH_{\mathrm{red}}\longrightarrow \cH_S,
    \label{eq:structured-Upsilon-map}
\end{equation}
by
\begin{equation}
    \Upsilon\,\iota^R_\alpha(v)
    :=\iota^S_\alpha(\ket{\mathrm{ref}_\alpha}\otimes v).
    \label{eq:structured-Upsilon-definition}
\end{equation}
Thus \(\Upsilon\) loads a reduced vector into a fixed representative copy of
the corresponding full block.

\begin{lemma}
\label{lem:structured-intertwining}
For every \(x\),
\begin{equation}
    \Xi U(x)
    =\bigl(I_{\cH_{\mathrm{mult}}}\otimes U_{\mathrm{red}}(x)\bigr)\Xi,
    \label{eq:structured-Xi-intertwining}
\end{equation}
and
\begin{equation}
    U(x)\Upsilon
    =\Upsilon U_{\mathrm{red}}(x).
    \label{eq:structured-Upsilon-intertwining}
\end{equation}
Equivalently,
\begin{equation}
    \Xi^\dagger
    \bigl(I_{\cH_{\mathrm{mult}}}\otimes U_{\mathrm{red}}(x)\bigr)
    \Xi
    =U(x),
    \label{eq:structured-Xi-compression}
\end{equation}
and
\begin{equation}
    \Upsilon^\dagger U(x)\Upsilon=U_{\mathrm{red}}(x).
    \label{eq:structured-Upsilon-compression}
\end{equation}
Taking adjoints gives the corresponding identities with \(U(x)^\dagger\)
and \(U_{\mathrm{red}}(x)^\dagger\), which are used in the compilers below.
Moreover,
\begin{equation}
\begin{aligned}
    \bigl(I_{\cH_{\mathrm{mult}}}\otimes U_{\mathrm{red}}(x)\bigr)
    \operatorname{Ran}\Xi&\subseteq\operatorname{Ran}\Xi,\\
    U(x)\operatorname{Ran}\Upsilon&\subseteq\operatorname{Ran}\Upsilon.
\end{aligned}
\label{eq:structured-range-invariance}
\end{equation}
Thus loading and unloading never project away any part of a legal state.
\end{lemma}

\begin{proof}
For \(m\in \cM_\alpha\) and \(v\in \cV_\alpha\), Eq.~\eqref{eq:structured-block-family}
gives
\[
    U(x)\,\iota^S_\alpha(m\otimes v)
    =\iota^S_\alpha(m\otimes U_\alpha(x)v).
\]
Applying \(\Xi\) gives the same vector as first applying \(\Xi\) and then
\(I\otimes U_{\mathrm{red}}(x)\), proving Eq.~\eqref{eq:structured-Xi-intertwining}.
The proof of Eq.~\eqref{eq:structured-Upsilon-intertwining} is identical, with
\(m=\ket{\mathrm{ref}_\alpha}\). The same blockwise calculation proves
Eq.~\eqref{eq:structured-range-invariance}. The compression identities then
follow by multiplying the intertwining identities by the adjoint maps; range
invariance ensures that the legal subspaces are preserved.
\end{proof}

\subsection{Multiplicity decoupling as a circuit compiler}

We now prove that removing multiplicities does not change the query number. We
do this in both directions: first lift a circuit from the reduced system to
the full system, and then restrict a full-system circuit to one representative
copy of each block.

\begin{proposition}
\label{prop:structured-lift-reduced}
Suppose there is an \(n\)-query clean circuit on \(\cH_{\mathrm{red}}\) such that
\begin{equation}
    \mathcal C_{\mathrm{red}}[U_{\mathrm{red}}](x)
    =\omega(x)U_{\mathrm{red}}(x)^\dagger.
    \label{eq:structured-reduced-exact}
\end{equation}
Then there is an \(n\)-query clean circuit on \(\cH_S\) such that
\begin{equation}
    \mathcal C_{\mathrm{full}}[U](x)
    =\omega(x)U(x)^\dagger.
    \label{eq:structured-full-lifted-exact}
\end{equation}
\end{proposition}

\begin{proof}
The reduced circuit is lifted to the physical system by a split--load--unload
wiring. The split map \(\Xi\) separates the passive multiplicity register from
the active reduced register. Whenever the reduced circuit requires a forward
call on the active register, the load isometry \(\Upsilon\) places that state
in a fixed representative physical copy, \(U(x)\) is applied once, and
\(\Upsilon^\dagger\) unloads it:
\[
    R \xrightarrow{\ \Upsilon\ } S
    \xrightarrow{\ U(x)\ } S
    \xrightarrow{\ \Upsilon^\dagger\ } R .
\]
Because \(\Upsilon^\dagger U(x)\Upsilon=U_{\mathrm{red}}(x)\), each reduced
forward call is implemented by exactly one call to the physical operation.
The passive multiplicity register bypasses all forward calls and is recombined with
\(\Xi^\dagger\) at the output. Hence the query number stays \(n\). On the
legal subspace, the implemented operator
is
\begin{equation}
    \mathcal C_{\mathrm{full}}[U](x)
    =\Xi^\dagger
    \bigl(I_{\cH_{\mathrm{mult}}}\otimes
    \mathcal C_{\mathrm{red}}[U_{\mathrm{red}}](x)\bigr)
    \Xi.
    \label{eq:structured-lift-composed}
\end{equation}
Using Eq.~\eqref{eq:structured-reduced-exact} and then
Eq.~\eqref{eq:structured-Xi-compression} gives
\[
    \mathcal C_{\mathrm{full}}[U](x)
    =\omega(x)\Xi^\dagger
    \bigl(I\otimes U_{\mathrm{red}}(x)^\dagger\bigr)\Xi
    =\omega(x)U(x)^\dagger.
\]
The fixed wiring and the reduced circuit return all auxiliary registers to
their initial states.
\end{proof}

\begin{proposition}
\label{prop:structured-reduce-full}
Suppose there is an \(n\)-query clean circuit on \(\cH_S\) such that
\begin{equation}
    \mathcal C_{\mathrm{full}}[U](x)
    =\omega(x)U(x)^\dagger.
    \label{eq:structured-full-exact}
\end{equation}
Then there is an \(n\)-query clean circuit on \(\cH_{\mathrm{red}}\) such that
\begin{equation}
    \mathcal C_{\mathrm{red}}[U_{\mathrm{red}}](x)
    =\omega(x)U_{\mathrm{red}}(x)^\dagger.
    \label{eq:structured-reduced-projected-exact}
\end{equation}
\end{proposition}

\begin{proof}
This is the reverse compiler. The reduced input is loaded into the chosen
representative physical copy by \(\Upsilon\). Whenever the full-system circuit
asks for one forward call, it is implemented through the split wiring
\[
    S \xrightarrow{\ \Xi\ } M R
    \xrightarrow{\ I_M\otimes U_{\mathrm{red}}(x)\ } M R
    \xrightarrow{\ \Xi^\dagger\ } S .
\]
By Eq.~\eqref{eq:structured-Xi-compression}, this module implements one
full-system call using exactly one reduced-system call. The final physical
register is unloaded with \(\Upsilon^\dagger\), and the query number is
unchanged. The
resulting reduced circuit satisfies
\begin{equation}
    \mathcal C_{\mathrm{red}}[U_{\mathrm{red}}](x)
    =\Upsilon^\dagger\mathcal C_{\mathrm{full}}[U](x)\Upsilon.
    \label{eq:structured-reduce-composed}
\end{equation}
Using Eq.~\eqref{eq:structured-full-exact} and then
Eq.~\eqref{eq:structured-Upsilon-compression},
\[
    \mathcal C_{\mathrm{red}}[U_{\mathrm{red}}](x)
    =\omega(x)\Upsilon^\dagger U(x)^\dagger\Upsilon
    =\omega(x)U_{\mathrm{red}}(x)^\dagger.
\]
\end{proof}

\begin{theorem}
\label{thm:structured-multiplicity-decoupling}
For every \(n\ge 0\), clean exact deterministic \(n\)-query inversion exists for
\(U(x)\) on \(\cH_S\) if and only if it exists for \(U_{\mathrm{red}}(x)\)
on \(\cH_{\mathrm{red}}\).  The same global phase function can be used on both
sides.  In particular, removing multiplicities does not change the minimum
query number.
\end{theorem}

\begin{proof}
For \(n\ge 1\), the two implications are exactly
Propositions~\ref{prop:structured-lift-reduced} and
\ref{prop:structured-reduce-full}.  For \(n=0\), the same two compilers contain
no query slots and reduce to the fixed wirings \(\Xi\) and
\(\Upsilon\), so the argument is unchanged.
\end{proof}

\subsection{Circuit primitives on the reduced system}

From now on we work only on the reduced system.  The block inclusions satisfy
\begin{equation}
    (\iota^R_\alpha)^\dagger\iota^R_\beta=\delta_{\alpha\beta}I_{\cV_\alpha},
    \qquad
    \sum_{\alpha\in\Lambda_{\mathrm{blk}}}
    \iota^R_\alpha(\iota^R_\alpha)^\dagger
    =I_{\cH_{\mathrm{red}}},
    \qquad
    U_{\mathrm{red}}(x)\iota^R_\alpha
    =\iota^R_\alpha U_\alpha(x).
    \label{eq:structured-reduced-block-relations}
\end{equation}

\begin{lemma}
\label{lem:structured-embed-block}
Let \(P_\alpha\) be a clean \(q\)-query circuit on \(\cV_\alpha\).  Then there is a
clean \(q\)-query circuit \(P_\alpha^\uparrow\) on \(\cH_{\mathrm{red}}\) such that
\begin{equation}
    P_\alpha^\uparrow[U_{\mathrm{red}}](x)\iota^R_\alpha
    =\iota^R_\alpha P_\alpha[U_\alpha](x)
    \label{eq:structured-embedded-block-output}
\end{equation}
for all \(x\).
\end{lemma}

\begin{proof}
Replace each \(\cV_\alpha\) register in \(P_\alpha\) by an
\(\cH_{\mathrm{red}}\) register, with \(\cV_\alpha\) identified with the
\(\alpha\) block. On the corresponding subspace, use the fixed gates of
\(P_\alpha\) unchanged, and let them act as the identity on its orthogonal
complement. If the input lies in the \(\alpha\) block,
Eq.~\eqref{eq:structured-reduced-block-relations} shows that every forward call
acts as \(U_\alpha(x)\), so the circuit reproduces \(P_\alpha\).

On any other block, the fixed gates act as the identity, so the \(q\) forward
calls act on the data register alone. Since \(U_{\mathrm{red}}(x)\) preserves
every block, the data remains outside the \(\alpha\) block and all auxiliary
registers are untouched. By linearity, the circuit is clean on all of
\(\cH_{\mathrm{red}}\), including coherent superpositions of different blocks,
and satisfies Eq.~\eqref{eq:structured-embedded-block-output}.
\end{proof}

\begin{definition}[Local block inverter]
\label{def:local-block-inverter}
A local block inverter for block \(\alpha\) is a clean circuit
\(\mathcal P_\alpha\) on \(\cV_\alpha\) using \(\ell_\alpha\) queries to
\(U_\alpha(x)\) and satisfying
\[
  \mathcal P_\alpha[U_\alpha](x)
  =\rho_\alpha(x)U_\alpha(x)^\dagger
\]
for all parameter values \(x\).  The primitive is not required to work for
every element of \(U(\cV_\alpha)\); when it is universal over \(U(\cV_\alpha)\),
Lemma~\ref{lem:universal-inverter-character} forces
\(\ell_\alpha+1=d_\alpha s_\alpha\) and
\(\rho_\alpha(x)=\det(U_\alpha(x))^{s_\alpha}\).
\end{definition}

The next circuit extracts the determinant phase from an active block by using
its one-dimensional top exterior power.

\begin{lemma}
\label{lem:structured-det-gadget}
For every \(\alpha\), there is a clean \(d_\alpha\)-query circuit
\(\mathcal D_\alpha^{(1)}\) on \(\cH_{\mathrm{red}}\) such that
\begin{equation}
    \mathcal D_\alpha^{(1)}[U_{\mathrm{red}}](x)
    =\det U_\alpha(x)\,I_{\cH_{\mathrm{red}}}.
    \label{eq:structured-det-gadget}
\end{equation}
\end{lemma}

\begin{proof}
Keep the data register untouched and prepare \(d_\alpha\) auxiliary registers
in the normalized antisymmetric state spanning
\(\bigwedge^{d_\alpha}\cV_\alpha\). Applying \(U_{\mathrm{red}}(x)\) once to
each register acts on this one-dimensional space as
\[
    \left.U_\alpha(x)^{\otimes d_\alpha}
    \right|_{\wedge^{d_\alpha}\cV_\alpha}
    =\det U_\alpha(x)\,I.
\]
The inverse preparation resets the auxiliary registers and leaves the scalar
phase \(\det U_\alpha(x)\) on the data register. Thus the circuit realizes
Eq.~\eqref{eq:structured-det-gadget} using \(d_\alpha\) forward calls.
\end{proof}

\begin{corollary}
\label{cor:structured-powered-det-gadget}
For every integer \(m\ge 1\), there is a clean \(md_\alpha\)-query circuit
\(\mathcal D_\alpha^{(m)}\) on \(\cH_{\mathrm{red}}\) such that
\begin{equation}
    \mathcal D_\alpha^{(m)}[U_{\mathrm{red}}](x)
    =\bigl(\det U_\alpha(x)\bigr)^m I_{\cH_{\mathrm{red}}}.
    \label{eq:structured-powered-det-gadget}
\end{equation}
\end{corollary}

\begin{proof}
Run the determinant circuit independently \(m\) times. Query numbers add and
scalar phases multiply.
\end{proof}

\subsection{Blockwise inversion as a synchronized branching circuit}

The following determinant construction is the universal-specialized
instance of the heterogeneous compiler.  Its local ingredient is a universal
\(U(d_\alpha)\) inverter on each active block.

We now assemble the full blockwise inversion on $\cH_{\mathrm{red}}$.  Assume
that, for every block $\alpha$, a clean exact \(\ell_\alpha\text{-query}\) universal inverter
\(\mathcal P_\alpha\) on $\cV_\alpha$ has been chosen.  Applying
Lemma~\ref{lem:universal-inverter-character} to each chosen block primitive
gives an integer
\begin{equation}
    s_\alpha=\frac{\ell_\alpha+1}{d_\alpha}
\end{equation}
and, after a fixed phase normalization,
\begin{equation}
    \mathcal P_\alpha[W]=(\det W)^{s_\alpha}W^\dagger
    \qquad\text{for every }W\in U(\cV_\alpha).
\end{equation}
Thus the determinant-synchronization construction may use
\begin{equation}
    \ell_\alpha+1=d_\alpha s_\alpha,
    \label{eq:structured-primitive-count}
\end{equation}
and the clean \(\ell_\alpha\text{-query}\) block primitive
\begin{equation}
    \mathcal P_\alpha[W]=(\det W)^{s_\alpha}W^\dagger .
    \label{eq:structured-block-primitive}
\end{equation}

We denote the determinant common phase by
\begin{equation}
    \zeta_{\mathrm{det}}(x):=
    \prod_{\alpha\in\Lambda_{\mathrm{blk}}}
    \bigl(\det U_\alpha(x)\bigr)^{s_\alpha}.
    \label{eq:structured-branch-phase}
\end{equation}

\begin{proposition}
\label{prop:structured-branch-protocol}
Fix a target block \(\nu\).  There is a clean circuit \(B_\nu\) on
\(\cH_{\mathrm{red}}\), using
\begin{equation}
    q_{\mathrm{det}}
    :=\sum_{\alpha\in\Lambda_{\mathrm{blk}}}(\ell_\alpha+1)-1
    =\sum_{\alpha\in\Lambda_{\mathrm{blk}}}d_\alpha s_\alpha-1
    \label{eq:structured-branch-query-count}
\end{equation}
queries, such that on inputs supported in block \(\nu\),
\begin{equation}
    B_\nu[U_{\mathrm{red}}](x)\iota^R_\nu
    =\zeta_{\mathrm{det}}(x)\iota^R_\nu U_\nu(x)^\dagger.
    \label{eq:structured-branch-output}
\end{equation}
\end{proposition}

\begin{proof}
On the branch labelled by \(\nu\), the active data register first passes through
the imported one-block inverter for \(\cV_\nu\), using
Lemma~\ref{lem:structured-embed-block}.  On \(\iota^R_\nu(\cV_\nu)\), this
produces
\[
    \bigl(\det U_\nu(x)\bigr)^{s_\nu}
    \iota^R_\nu U_\nu(x)^\dagger
\]
and uses \(\ell_\nu=d_\nu s_\nu-1\) forward calls. Determinant circuits
for the remaining blocks supply the missing determinant charges:
\(\mathcal D_\alpha^{(s_\alpha)}\) contributes
\((\det U_\alpha(x))^{s_\alpha}\) and uses \(d_\alpha s_\alpha\) traversals
for every \(\alpha\ne\nu\). Since these loops act as scalars on the data
register, their order does not matter. Multiplying all scalar factors gives
\(\zeta_{\mathrm{det}}(x)\), hence Eq.~\eqref{eq:structured-branch-output}.
The query count is
\[
    \ell_\nu+\sum_{\alpha\ne\nu}d_\alpha s_\alpha
    =(d_\nu s_\nu-1)+\sum_{\alpha\ne\nu}d_\alpha s_\alpha
    =\sum_\alpha d_\alpha s_\alpha-1.
\]
\end{proof}

\begin{lemma}
\label{lem:structured-controlled-sum}
Let \(q_{\mathrm{com}}\ge 0\).  Suppose that, for each block \(\nu\), there is a clean
\(q_{\mathrm{com}}\)-query circuit \(C^{(\nu)}\) on \(\cH_{\mathrm{red}}\) and an operator
family \(T_\nu(x)\in\cB(\cV_\nu)\) such that
\begin{equation}
    C^{(\nu)}[U_{\mathrm{red}}](x)\iota^R_\nu
    =\iota^R_\nu T_\nu(x).
    \label{eq:structured-controlled-input}
\end{equation}
Then there is a single clean \(q_{\mathrm{com}}\)-query circuit \(C\) on
\(\cH_{\mathrm{red}}\) satisfying
\begin{equation}
    C[U_{\mathrm{red}}](x)\iota^R_\nu
    =\iota^R_\nu T_\nu(x)
    \qquad\text{for every }\nu.
    \label{eq:structured-controlled-output}
\end{equation}
Equivalently,
\begin{equation}
    C[U_{\mathrm{red}}](x)
    =\bigoplus_{\nu\in\Lambda_{\mathrm{blk}}}T_\nu(x).
    \label{eq:structured-controlled-direct-sum}
\end{equation}
\end{lemma}

\begin{proof}
Attach a label register
\(L=\operatorname{span}\{\ket{\nu}:\nu\in\Lambda_{\mathrm{blk}}\}\).  First apply
the fixed block-label extraction unitary
\begin{equation}
    \ket{0}_L\otimes\sum_\nu\iota^R_\nu x_\nu
    \longmapsto
    \sum_\nu\ket{\nu}_L\otimes\iota^R_\nu x_\nu.
    \label{eq:structured-controlled-label}
\end{equation}
Because the sectors \(\iota^R_\nu(\cV_\nu)\) are orthogonal, this is an isometry
and hence can be completed to a unitary.

Now run a coherent if--else circuit controlled by \(L\). At query slot
\(r=1,\ldots,q_{\mathrm{com}}\), branch \(\nu\) applies exactly the fixed pre-query and
post-query gates that the circuit \(C^{(\nu)}\) would apply at slot \(r\).  In
formula, each fixed gate has the controlled form
\[
    \sum_\nu \ket{\nu}\!\bra{\nu}_L\otimes G_r^{(\nu)},
\]
after all branch ancilla spaces have been embedded into one sufficiently
large common ancilla. The forward call in the middle is the same single call
to \(U_{\mathrm{red}}(x)\); if different branches want to query different
work registers, controlled swaps move the desired register into the common
input port before the call and move it back afterwards. The branch control
acts only on fixed gates and swap networks; the unknown operation itself is
called once, without control, on the common port.

The label \(\nu\) is never changed, so an input initially in sector \(\nu\)
follows precisely the circuit \(C^{(\nu)}\).  The branch output lies again in
sector \(\nu\) by Eq.~\eqref{eq:structured-controlled-input}, so the inverse of
Eq.~\eqref{eq:structured-controlled-label} erases the label coherently.  This
proves Eq.~\eqref{eq:structured-controlled-output}.
\end{proof}

\begin{theorem}
\label{thm:structured-reduced-construction}
For every block $\alpha$, fix the normalized clean universal inverter
$\mathcal P_\alpha$ specified in Eqs.~\eqref{eq:structured-primitive-count}
and~\eqref{eq:structured-block-primitive}. Then
there is a clean circuit on $\cH_{\mathrm{red}}$ using
\begin{equation}
    q_{\mathrm{det}}
    =\sum_{\alpha\in\Lambda_{\mathrm{blk}}}(\ell_\alpha+1)-1
    =\sum_{\alpha\in\Lambda_{\mathrm{blk}}}d_\alpha s_\alpha-1
    \label{eq:structured-final-query-count}
\end{equation}
queries and implementing
\begin{equation}
    \mathcal C_{\mathrm{red}}^{\mathrm{blk}}[U_{\mathrm{red}}](x)
    =\zeta_{\mathrm{det}}(x)U_{\mathrm{red}}(x)^\dagger,
    \label{eq:structured-final-reduced-output}
\end{equation}
where \(\zeta_{\mathrm{det}}(x)\) is given by Eq.~\eqref{eq:structured-branch-phase}.
\end{theorem}

\begin{proof}
If \(q_{\mathrm{det}}=0\), then
\(\sum_\alpha d_\alpha s_\alpha=1\).  Since each \(d_\alpha,s_\alpha\ge 1\),
there is only one block and \(d_\alpha=s_\alpha=1\).  In that case
\(U_{\mathrm{red}}(x)\) is a scalar unitary and
\(\zeta_{\mathrm{det}}(x)U_{\mathrm{red}}(x)^\dagger=I\), so the identity circuit is the
required zero-query protocol.

Assume now \(q_{\mathrm{det}}\ge 1\).  For each target block \(\nu\),
Proposition~\ref{prop:structured-branch-protocol} gives an
\(q_{\mathrm{det}}\)-query branch circuit with
\[
    B_\nu[U_{\mathrm{red}}](x)\iota^R_\nu
    =\zeta_{\mathrm{det}}(x)\iota^R_\nu U_\nu(x)^\dagger.
\]
Apply Lemma~\ref{lem:structured-controlled-sum} with
\(T_\nu(x)=\zeta_{\mathrm{det}}(x)U_\nu(x)^\dagger\).  The resulting synchronized
branching circuit satisfies, for every \(\nu\),
\[
    \mathcal C_{\mathrm{red}}^{\mathrm{blk}}[U_{\mathrm{red}}](x)\iota^R_\nu
    =\zeta_{\mathrm{det}}(x)\iota^R_\nu U_\nu(x)^\dagger.
\]
Since the block sectors form an orthogonal direct sum of
\(\cH_{\mathrm{red}}\), this is Eq.~\eqref{eq:structured-final-reduced-output}.
\end{proof}

\paragraph*{Relation to the Letter.}
Theorem~\ref{thm:structured-multiplicity-decoupling} proves the Wedderburn-reduction theorem of the Letter, including the zero-query convention. Theorem~\ref{thm:structured-reduced-construction}, followed by Proposition~\ref{prop:structured-lift-reduced}, gives the determinant specialization of its phase-coherent blockwise compiler.

\subsection{Block-phase barycentric synchronization}
\label{appendix:block-phase-barycentric}

The determinant-synchronization circuit above makes every branch carry the
same full determinant product \(\zeta_{\mathrm{det}}(x)\). To recombine the
branches, Lemma~\ref{lem:structured-controlled-sum} requires only two things:
the branches must have the same query length and the same scalar phase. We
record this more general condition before giving
its determinant trace-vector specialization.

\begin{theorem}[Phase-coherent blockwise synchronization]
\label{thm:block-phase-general-sync}
Let
\[
    \cH_{\mathrm{red}}
    =
    \bigoplus_{\alpha\in\Lambda_{\mathrm{blk}}}\cV_\alpha,
    \qquad
    U_{\mathrm{red}}(x)
    =
    \bigoplus_{\alpha\in\Lambda_{\mathrm{blk}}}U_\alpha(x).
\]
For every block \(\alpha\), suppose there is a clean one-block inverter
\(\mathcal P_\alpha\) using \(\ell_\alpha\) forward calls and satisfying
\[
    \mathcal P_\alpha[U_\alpha](x)
    =
    \rho_\alpha(x)U_\alpha(x)^\dagger .
\]
No determinant form is assumed for \(\rho_\alpha\); the theorem takes
\((\ell_\alpha,\rho_\alpha)\) as the supplied local primitive data.
Suppose also that there is a finite set $\Gamma$ of clean scalar gadgets
\(\mathcal S_\gamma\), where \(\gamma\) is a scalar-gadget label, kept
separate from query-slot indices, such
that \(\mathcal S_\gamma\) uses \(w_\gamma\) forward calls on
\(\cH_{\mathrm{red}}\) and implements
\[
    \mathcal S_\gamma[U_{\mathrm{red}}](x)
    =
    \eta_\gamma(x)I_{\cH_{\mathrm{red}}}.
\]
If there exist nonnegative integers \(n_{\alpha\gamma}\), a common query number
\(q_{\mathrm{com}}\), and a common phase \(\zeta(x)\) such that, for every block
\(\alpha\),
\begin{align}
    \ell_\alpha+\sum_{\gamma\in\Gamma}n_{\alpha\gamma}w_\gamma&=q_{\mathrm{com}},
    \label{eq:block-phase-general-cost}\\
    \rho_\alpha(x)
    \prod_{\gamma\in\Gamma}\eta_\gamma(x)^{n_{\alpha\gamma}}
    &=\zeta(x)
    \qquad\text{for all }x,
    \label{eq:block-phase-general-phase}
\end{align}
then there is a clean \(q_{\mathrm{com}}\)-query circuit on \(\cH_{\mathrm{red}}\) that
implements
\[
    \zeta(x)U_{\mathrm{red}}(x)^\dagger .
\]
Consequently, the reduced-to-full compiler gives a clean \(q_{\mathrm{com}}\)-query circuit
on \(\cH_S\) implementing \(\zeta(x)U(x)^\dagger\).
\end{theorem}

\begin{proof}
Fix a target block \(\nu\).  Import the one-block inverter
\(\mathcal P_\nu\) to the reduced system using
Lemma~\ref{lem:structured-embed-block}.  On inputs supported in
\(\iota^R_\nu(\cV_\nu)\), this gives
\[
    \rho_\nu(x)\iota^R_\nu U_\nu(x)^\dagger
\]
and uses \(\ell_\nu\) queries.  Then run each scalar gadget \(\mathcal S_\gamma\)
exactly \(n_{\nu\gamma}\) times.  The scalar gadgets act as multiples of the
identity on the whole reduced system, so on the \(\nu\)-sector the resulting
branch circuit has query count
\[
    \ell_\nu+\sum_{\gamma\in\Gamma}n_{\nu\gamma}w_\gamma=q_{\mathrm{com}}
\]
and output
\[
    \rho_\nu(x)\prod_{\gamma\in\Gamma}\eta_\gamma(x)^{n_{\nu\gamma}}
    \iota^R_\nu U_\nu(x)^\dagger
    =
    \zeta(x)\iota^R_\nu U_\nu(x)^\dagger .
\]
Thus every block has a clean branch circuit with the same query number \(q_{\mathrm{com}}\)
and sector output
\[
    T_\nu(x)=\zeta(x)U_\nu(x)^\dagger .
\]
Lemma~\ref{lem:structured-controlled-sum} coherently combines these branches
into a single reduced-system circuit implementing
\[
    \bigoplus_{\nu\in\Lambda_{\mathrm{blk}}}
    \zeta(x)U_\nu(x)^\dagger
    =
    \zeta(x)U_{\mathrm{red}}(x)^\dagger .
\]
The final full-system statement follows from
Proposition~\ref{prop:structured-lift-reduced}.
\end{proof}

When the residual phases and scalar gadgets are characters,
\[
  \rho_\alpha(x)=e^{ir_\alpha\cdot x},
  \qquad
  \eta_\gamma(x)=e^{ih_\gamma\cdot x},
\]
the general synchronization conditions become the weighted integer program
\[
\begin{aligned}
  \ell_\alpha+
  \sum_{\gamma\in\Gamma} n_{\alpha\gamma}w_\gamma&=q_{\mathrm{com}},\\
  r_\alpha+
  \sum_{\gamma\in\Gamma} n_{\alpha\gamma}h_\gamma&=r_\star,
\end{aligned}
\]
with common \(q_{\mathrm{com}}\) and \(r_\star\).  The trace-vector theorem in
the main text is the determinant specialization
\(r_\alpha=s_\alpha\tau_\alpha\), \(h_\beta=\tau_\beta\), and
\(w_\beta=d_\beta\).

\begin{lemma}[Closed scalar from a local inverter]
\label{lem:scalar-from-local-inverter}
Let \(\mathcal P_\alpha\) be a local block inverter using
\(\ell_\alpha\) calls and leaving the phase \(\rho_\alpha(x)\). Then the reduced system
admits a clean scalar gadget \(\mathcal S_\alpha\) using \(\ell_\alpha+1\)
queries and implementing
\[
  \mathcal S_\alpha[U_{\mathrm{red}}](x)
  =\rho_\alpha(x)I_{\cH_{\mathrm{red}}}.
\]
\end{lemma}

\begin{proof}
Bypass the data register.  Prepare an auxiliary register in a fixed state
inside the sector \(\iota^R_\alpha(\cV_\alpha)\).  Run the imported circuit
\(\mathcal P_\alpha^\uparrow\) on this auxiliary register and then send the same
auxiliary register once through the reduced operation. On the legal auxiliary
subspace the combined operation is
\[
  U_\alpha(x)\,\mathcal P_\alpha[U_\alpha](x)
  =\rho_\alpha(x)I_{\cV_\alpha}.
\]
Thus the auxiliary state is restored, and the bypassed data register acquires
the scalar phase \(\rho_\alpha(x)\). The circuit uses \(\ell_\alpha+1\) calls, and no
explicit formula for \(\rho_\alpha\) is needed.
\end{proof}

\begin{corollary}[Automatic blockwise completion]
\label{cor:heterogeneous-fallback}
Given local block inverters \(\mathcal P_\alpha\) with query numbers
\(\ell_\alpha\) and phases \(\rho_\alpha(x)\), there is a clean exact
deterministic protocol using
\[
  q_{\mathrm{auto}}=
  \sum_{\alpha\in\Lambda_{\mathrm{blk}}}(\ell_\alpha+1)-1
\]
queries and implementing
\[
  \left(\prod_{\alpha\in\Lambda_{\mathrm{blk}}}\rho_\alpha(x)\right)
  U(x)^\dagger .
\]
\end{corollary}

\begin{proof}
Use Theorem~\ref{thm:block-phase-general-sync} with scalar gadgets
\(\mathcal S_\beta\) from Lemma~\ref{lem:scalar-from-local-inverter}.  On branch
\(\alpha\), run the local inverter \(\mathcal P_\alpha\) and run
\(\mathcal S_\beta\) once for every \(\beta\ne\alpha\). The total query number is
\[
  \ell_\alpha+
  \sum_{\beta\ne\alpha}(\ell_\beta+1)
  =\sum_\beta(\ell_\beta+1)-1,
\]
and the branch phase is
\[
  \rho_\alpha(x)\prod_{\beta\ne\alpha}\rho_\beta(x)
  =\prod_\beta\rho_\beta(x).
\]
The reduced-to-full compiler then gives the full-system statement.
\end{proof}

Let $q_{\rm univ}(d)$ denote the query number of the deterministic exact
protocol of Ref.~\smcite{2} for inverting an arbitrary
$d$-dimensional unitary. For this protocol,
$q_{\rm univ}(d)=\mathcal O(d^2)$. Choosing this protocol as the local inverter
gives $\ell_\alpha=q_{\rm univ}(d_\alpha)$. Since every
$d_\alpha\ge1$, Corollary~\ref{cor:heterogeneous-fallback} therefore implies
the general constructive bound
\[
  q_{\rm auto}=\mathcal O\!\left(\sum_{\alpha\in\Lambda_{\rm blk}}
  d_\alpha^2\right).
\]

For each active block, define its determinant character and trace vector by
\begin{equation}
  D_\beta(x):=\det U_\beta(x)=e^{i\tau_\beta\cdot x},
  \qquad
  (\tau_\beta)_j=\operatorname{tr}H_{j,\beta}.
  \label{eq:block-trace-vector-supp}
\end{equation}

\begin{corollary}[Determinant trace-vector synchronization]
\label{cor:determinant-trace-sync}
Fix the normalized universal block primitives of
Eqs.~\eqref{eq:structured-primitive-count}
and~\eqref{eq:structured-block-primitive}. Suppose there exist vectors
\(v^{(\alpha)}\in\ZZ_{\ge0}^{\Lambda_{\mathrm{blk}}}\), an integer \(q\ge0\),
and a vector \(\tau_\star\in\RR^p\) such that, for every \(\alpha\),
\[
\begin{aligned}
  v^{(\alpha)}_\alpha&\ge s_\alpha,\\
  \sum_\beta d_\beta v^{(\alpha)}_\beta&=q+1,\\
  \sum_\beta v^{(\alpha)}_\beta\tau_\beta&=\tau_\star.
\end{aligned}
\]
Then a clean exact deterministic \(q\)-query protocol implements
\(e^{i\tau_\star\cdot x}U(x)^\dagger\).
\end{corollary}

\begin{proof}
Recall from Eq.~\eqref{eq:block-trace-vector-supp} that
\(D_\beta(x)=\exp(i\tau_\beta\cdot x)\).
Use Theorem~\ref{thm:block-phase-general-sync} with the determinant-charge
block primitives
\[
    \ell_\alpha=d_\alpha s_\alpha-1,
    \qquad
    \rho_\alpha(x)=D_\alpha(x)^{s_\alpha},
\]
and with scalar gadgets given by the determinant circuits
\(\mathcal D_\beta^{(1)}\).  For the gadget labelled by \(\beta\),
\[
    w_\beta=d_\beta,
    \qquad
    \eta_\beta(x)=D_\beta(x).
\]
Given the determinant inventory \(v^{(\alpha)}\), set
\[
    n_{\alpha\beta}
    :=
    v^{(\alpha)}_\beta-s_\alpha\delta_{\alpha\beta}.
\]
The lower bound \(v^{(\alpha)}_\alpha\ge s_\alpha\) makes all
\(n_{\alpha\beta}\) nonnegative.  The query count on branch \(\alpha\) is
\[
\begin{aligned}
    \ell_\alpha+\sum_\beta n_{\alpha\beta}d_\beta
    &=
    (d_\alpha s_\alpha-1)
    +
    \sum_\beta
    d_\beta\bigl(v^{(\alpha)}_\beta
    -s_\alpha\delta_{\alpha\beta}\bigr)\\
	    &=
	    \sum_\beta d_\beta v^{(\alpha)}_\beta-1
	    =
	    q .
	\end{aligned}
	\]
The scalar phase on the same branch is
\[
    D_\alpha(x)^{s_\alpha}
    \prod_\beta
    D_\beta(x)^{
    v^{(\alpha)}_\beta-s_\alpha\delta_{\alpha\beta}}
    =
    \prod_\beta D_\beta(x)^{v^{(\alpha)}_\beta}
    =
    \exp\!\left(
    i x\cdot\sum_\beta v^{(\alpha)}_\beta\tau_\beta
    \right)
    =
    e^{i\tau_\star\cdot x}.
\]
Therefore Theorem~\ref{thm:block-phase-general-sync} applies with
\(q_{\mathrm{com}}=q\) and \(\zeta(x)=e^{i\tau_\star\cdot x}\).
\end{proof}

\begin{corollary}
\label{cor:traceless-block-sharing}
Assume that every active block generator is traceless,
\[
    \operatorname{tr}H_{i,\alpha}=0
    \qquad
    \forall\,i,\alpha .
\]
Then \(D_\alpha(x)=1\) for all blocks and all \(x\), so the trace-vector
condition in Corollary~\ref{cor:determinant-trace-sync} is automatic.
For the fixed local primitives above, and allowing padding only by the
determinant gadgets $\mathcal D_\beta^{(1)}$, let the padding semigroup be
\[
    \mathsf P_{\mathrm{blk}}
    :=
    \left\{
	    \sum_{\beta\in\Lambda_{\mathrm{blk}}}c_\beta d_\beta:
	    c_\beta\in\ZZ_{\ge0}
    \right\},
\]
and
\[
    q_\star
    :=
    \min\left\{
    q\ge0:\;
    q+1-d_\alpha s_\alpha\in\mathsf P_{\mathrm{blk}}
    \ \forall\,\alpha
    \right\}.
\]
Then \(q_\star\) queries suffice. Within this determinant-padding construction,
$q_\star$ is the minimum common branch length. In particular, if all active
blocks have the same dimension \(d\), then
\[
    q_\star=d\max_\alpha s_\alpha-1.
\]
For the universal protocol of Ref.~\smcite{2}, equal-dimensional
blocks use the same chosen primitive and hence the same $s_\alpha$. Thus $r$
equal-dimensional, not necessarily equivalent, traceless active blocks share
one block budget $q_{\rm univ}(d)$, instead of the determinant-sum budget
$r\bigl(q_{\rm univ}(d)+1\bigr)-1$.
\end{corollary}

\begin{proof}
For each block,
\[
    D_\alpha(x)
    =
    \det e^{iH_\alpha(x)}
    =
    \exp\!\bigl(i\,\operatorname{tr}H_\alpha(x)\bigr)
    =
    \exp(i\,\tau_\alpha\cdot x).
\]
If every generator \(H_{i,\alpha}\) is traceless, then
\(\tau_\alpha=0\) and hence \(D_\alpha(x)=1\) for every \(x\).  The trace-vector
condition in Corollary~\ref{cor:determinant-trace-sync} is therefore automatic.

It remains only to make all branch query lengths equal.  Branch \(\alpha\)
already has length \(d_\alpha s_\alpha-1\).  Adding
\(n_{\alpha\beta}\) copies of the determinant gadget for block \(\beta\)
adds \(\sum_\beta n_{\alpha\beta}d_\beta\) queries and no phase.
Thus, within this determinant-padding construction, a common query count \(q\)
is feasible exactly when
\[
    q+1-d_\alpha s_\alpha
    =
    \sum_\beta n_{\alpha\beta}d_\beta
    ,\qquad
    n_{\alpha\beta}\in\ZZ_{\ge0}
\]
for every \(\alpha\).  Minimizing over such \(q\) gives \(q_\star\), and
Corollary~\ref{cor:determinant-trace-sync} gives \(q_\star\) queries.
Writing $g=\gcd\{d_\beta:\beta\in\Lambda_{\mathrm{blk}}\}$, feasibility
requires all initial branch lengths $d_\alpha s_\alpha-1$ to be congruent
modulo $g$. In particular, if every available padding gadget has even length,
all initial branch lengths must have the same parity. The semigroup condition
above is the full criterion; parity alone need not be sufficient.

If all active dimensions are equal to \(d\), then each
\(d_\alpha s_\alpha\) is a multiple of \(d\), and choosing
\[
    q=d\max_\alpha s_\alpha-1
\]
pads every branch by an integer number of \(d\)-query determinant gadgets.
No smaller common \(q\) is possible within this construction because it must be at least
\(\max_\alpha d s_\alpha-1\).  This proves the equal-dimensional claim.  For the
universal protocol of Ref.~\smcite{2}, equal dimensions give the same
chosen primitive and hence a common value of \(s_\alpha\), yielding the
displayed shared budget.
No claim of global query optimality is made by this padding argument.
\end{proof}

\section{Worked constructions}\label{appendix:worked-examples}

\subsection{Three-slot spectral circuit}

Consider the distinct spectrum
\[
\Lambda=\{-5,-4,1,3,5\}.
\]
Here \(K=5\), so the universal \(K-1\) scalar-block routing uses \(4\) queries.
However, the matrix
\begin{equation}\label{eq:example-N}
N=
\begin{pmatrix}
1&0&0&0&2\\
0&1&0&1&1\\
1&0&1&1&0\\
1&0&2&0&0\\
2&0&0&0&1
\end{pmatrix}
\end{equation}
satisfies
\[
N\mathbf{1}_5=3\mathbf{1}_5,
\qquad
N\bm{\lambda}=-\bm{\lambda},
\qquad
\bm{\lambda}=(-5,-4,1,3,5)^T.
\]
Thus the spectrum is \(3\)-barycentrically symmetric around \(m=0\).

One corresponding choice of triples is
\[
\begin{aligned}
-5 &\longmapsto (-5,5,5),\\
-4 &\longmapsto (-4,3,5),\\
1 &\longmapsto (-5,1,3),\\
3 &\longmapsto (-5,1,1),\\
5 &\longmapsto (-5,-5,5).
\end{aligned}
\]
Each input eigenvalue together with its three assigned counterweights has
barycenter \(0\), and the construction above therefore yields an exact
\(3\)-slot inversion circuit.

To see that \(q=2\) is impossible, note that
\[
\Sigma_2(\Lambda)
=
\{-10,-9,-8,-4,-3,-2,-1,0,1,2,4,6,8,10\},
\]
and a direct inspection shows that
\[
\bigcap_{\lambda\in\Lambda}\bigl(\lambda+\Sigma_2(\Lambda)\bigr)=\varnothing.
\]
Hence no real number \(c\) satisfies \(c-\Lambda\subseteq \Sigma_2(\Lambda)\).
The case \(q=0\) is impossible because \(K>1\).  For \(q=1\), the sumset
condition would require \(c-\Lambda=\Lambda\).  The extreme pair gives
\(c=(-5)+5=0\), whereas the next pair gives \(c=(-4)+3=-1\), so no common
\(c\) exists.  Thus \(q=0,1,2\) are all excluded, and the minimum query number
for this spectrum is \(3\).
The speedup over the universal \(4\)-slot construction is thus witnessed by a
row-sum-\(3\) scalar-block synchronization matrix instead of the universal
construction \(J-I_5\) of row sum \(4\).

\begin{figure*}[t]
\centering
\includegraphics[width=0.8\linewidth]{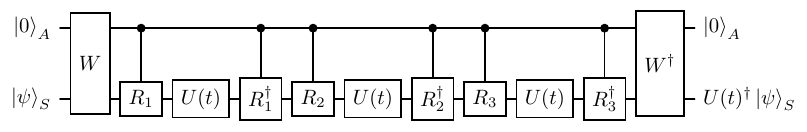}
\caption{Three-query scalar-block routing circuit realising the
matrix witness~\eqref{eq:example-N} for $\Lambda=\{-5,-4,1,3,5\}$, where
the optimum \(3\) improves on the universal bound \(K-1=4\).  The encoder
$W$ writes the eigenvalue label of the input onto the ancilla and compresses
the system register onto the representative subspace.  Each slot then routes
the query register, conditioned on that label, into the partner eigenspace
$\mu_{j,s}$ selected at slot $s$ of the ordered tuple (gate $R_s$), applies
$U(t)$ once, and reverses the routing ($R_s^\dagger$). Every input
eigenvalue has barycenter zero, so the three acquired phases sum to
$-\lambda_j t$; after the decoder $W^\dagger$ the system carries
$U(t)^\dagger\ket{\psi}_S$ and the ancilla returns to $\ket{0}_A$. The diagram
is schematic: $A$ collects the spectral-label, multiplicity, and work
registers used in the explicit construction.}
\label{fig:three-slot-spectral}
\end{figure*}

\section{Application scaling certificates}
\label{appendix:application-scalings}

This section proves the query bounds quoted in the application section of the
Letter. For the finite-$n$ values in Fig.~2, we use
$q_{\rm univ}(1)=0$ and
$q_{\rm univ}(d)=d\lceil\pi/[2\arcsin(1/d)]\rceil-1$ for $d\ge2$;
hence $q_{\rm univ}(d)=\mathcal O(d^2)$~\smcite{2}. The
arbitrary-$X$ curves are exact
finite-$n$ evaluations of constructive upper bounds, while the two fixed-basis
curves use the optima proved below. Fixed basis changes, encoders, and routing
gates are not charged.

\subsection{Multimode Tavis--Cummings dynamics}
\label{appendix:TC-scaling}

Consider $n$ identical two-level emitters and $M$ resolved bosonic modes, with
\begin{equation}
 H_{\rm TC}
 =\omega_0J_z+\sum_{\mu=1}^{M}\omega_\mu a_\mu^\dagger a_\mu
 +\sum_{\mu=1}^{M}
 \left(g_\mu J_+a_\mu+g_\mu^*J_-a_\mu^\dagger\right).
 \label{eq:supp-TC-H}
\end{equation}
The unknown scalar parameters are $\omega_0,\omega_\mu\in\RR$ and $g_\mu\in\mathbb C$.
Writing $g_\mu=u_\mu+iv_\mu$ makes the known Hermitian generator pairs
explicit:
\begin{equation}
 g_\mu J_+a_\mu+g_\mu^*J_-a_\mu^\dagger
 =u_\mu\left(J_+a_\mu+J_-a_\mu^\dagger\right)
 +iv_\mu\left(J_+a_\mu-J_-a_\mu^\dagger\right).
 \label{eq:supp-TC-real-generators}
\end{equation}
Thus the family fits the real-coupling model of the Letter. It is invariant
under emitter permutations and conserves
\begin{equation}
 \mathcal N=J_z+\frac n2 I+\sum_{\mu=1}^{M}a_\mu^\dagger a_\mu.
 \label{eq:supp-TC-number}
\end{equation}
Indeed, $[\mathcal N,J_+a_\mu]=J_+a_\mu-J_+a_\mu=0$, and the adjoint term
commutes for the same reason.

\begin{proposition}[Tavis--Cummings query separation]
\label{prop:supp-TC-scaling}
Fix $M$ and an excitation cap $N$. For $n\ge2N$, the restriction of
Eq.~\eqref{eq:supp-TC-H} to the spectral subspace
$\mathcal H_{\le N}:=\mathbf 1_{\{\mathcal N\le N\}}\mathcal H$
admits a clean exact inverse obtained from coarse invariant sectors. Here the
restricted joint emitter--mode evolution is the forward operation supplied to
the protocol; fixed Schur transforms and routing gates are not counted. The
query number satisfies
\begin{equation}
 q_{\rm TC}
 \le
 \sum_{j=0}^{N}
 \left[q_{\rm univ}(d_j)+1\right]-1,
 \qquad
 d_j=\binom{M+N-j+1}{N-j}.
 \label{eq:supp-TC-query}
\end{equation}
Hence $q_{\rm TC}=\mathcal O_{M,N}(1)$. The same supported Hilbert space has
dimension $\mathcal O(n^N)$, so treating its unitary as arbitrary gives the
universal $\mathcal O(n^{2N})$ benchmark.
\end{proposition}

\begin{proof}
Schur--Weyl duality decomposes the emitter space as~\smcite{3}
\begin{equation}
 (\mathbb C^2)^{\otimes n}
 \cong
 \bigoplus_{j=0}^{\lfloor n/2\rfloor}
 \mathcal M_j\otimes V_{S_j},
 \qquad S_j=\frac n2-j,
 \label{eq:supp-TC-Schur}
\end{equation}
where $\mathcal M_j$ is a permutation multiplicity and $V_{S_j}$ is the
active spin representation. The collective spin operators act trivially on
$\mathcal M_j$, so the Hamiltonian itself has the block form
\begin{equation*}
 H_{\rm TC}
 \cong
 \bigoplus_{j=0}^{\lfloor n/2\rfloor}
 I_{\mathcal M_j}\otimes H_{\rm TC}^{(S_j)},
 \qquad
 H_{\rm TC}^{(S)}
 =\omega_0J_z^{(S)}+\sum_{\mu=1}^{M}\omega_\mu a_\mu^\dagger a_\mu
 +\sum_{\mu=1}^{M}
 \left(g_\mu J_+^{(S)}a_\mu+g_\mu^*J_-^{(S)}a_\mu^\dagger\right),
\end{equation*}
where $J_z^{(S)}$ and $J_\pm^{(S)}$ are the spin-$S$ matrices. Thus every
copy labelled by $\mathcal M_j$ undergoes the same joint spin--mode evolution.
Write the magnetic label as $m=-S_j+r$, with
$r\ge0$. The emitter excitation number on this vector is
\begin{equation}
 m+\frac n2=j+r.
 \label{eq:supp-TC-spin-excitation}
\end{equation}
Consequently, sector $j$ can occur in $\mathcal H_{\le N}$ only when $j\le N$.
For fixed $j$, a reduced basis vector is labelled by
$(r,m_1,\ldots,m_M)\in\ZZ_{\ge0}^{M+1}$. Its exact total excitation is
\begin{equation}
 L=j+r+\sum_{\mu=1}^{M}m_\mu.
 \label{eq:supp-TC-exact-excitation}
\end{equation}
The Hamiltonian preserves $L$, so the fixed-$j$ space is itself a direct sum
of exact-excitation sectors. The sector $(j,L)$ has dimension
\begin{equation}
 d_{j,L}=\binom{M+L-j}{L-j},
 \qquad j\le L\le N.
 \label{eq:supp-TC-exact-sector-dimension}
\end{equation}
The condition $n\ge2N$ ensures that the upper end of the spin ladder never
truncates these labels. For the stated compact bound, group the sectors
$L=j,\ldots,N$ into one coarse invariant sector. Its dimension is
\begin{equation}
 d_j
 =\sum_{L=j}^{N}d_{j,L}
 =\binom{M+N-j+1}{N-j}.
 \label{eq:supp-TC-coarse-dimension}
\end{equation}
The Wedderburn-reduction theorem therefore removes the repeated
$\mathcal M_j$ copies. The universal inversion protocol of
Ref.~\smcite{2} on each coarse
sector, followed by
Corollary~\ref{cor:heterogeneous-fallback}, gives
Eq.~\eqref{eq:supp-TC-query}. All $d_j$ depend only on $M,N$, proving the first
scaling claim. Resolving the exact-excitation sectors separately supplies a
finer block list and can only improve this certified upper bound.

For the universal comparison, choose $k$ excited emitters and at most
$N-k$ bosons. The supported dimension is
\begin{equation}
 D_{\rm TC}(n,M,N)
 =\sum_{k=0}^{N}\binom nk\binom{M+N-k}{N-k}.
 \label{eq:supp-TC-ambient-dimension}
\end{equation}
For fixed $M,N$, the $k=N$ term is $\binom nN$ and all other terms have lower
polynomial degree, so $D_{\rm TC}=\mathcal O(n^N)$. Universal exact inversion on
an arbitrary $D_{\rm TC}$-dimensional unitary has query complexity
$\mathcal O(D_{\rm TC}^2)=\mathcal O(n^{2N})$.
\end{proof}

The proposition bounds the coherent backward evolution of the joint
emitter--mode system used in the OTOC setting of the Letter. It assumes that
all resolved modes remain available to the circuit.
Cavity loss, spontaneous emission, and other dissipative processes are not
reversed and lie outside this ideal unitary query model.

\subsection{Echo-verified collective dynamics}
\label{appendix:collective-blocks}

Let $J_\mu=\frac12\sum_{r=1}^{n}\sigma_r^\mu$ be the collective-spin operator along $\mu\in\{x,y,z\}$, where $\sigma_r^\mu$ is the Pauli operator on spin $r$. We consider
\begin{equation}
 H_{\rm col}
 =\sum_{\mu=x,y,z}h_\mu J_\mu
 +\frac1n\sum_{\mu\le\nu}\chi_{\mu\nu}
 \left(J_\mu J_\nu+J_\nu J_\mu\right),
 \label{eq:supp-collective-family}
\end{equation}
where $h_\mu,\chi_{\mu\nu}\in\RR$ are the unknown scalar parameters. This includes linear collective fields and
one- and two-axis twisting subfamilies. Every term commutes with the emitter
permutation action.

\begin{proposition}[Collective-family query bounds]
\label{prop:supp-collective-scaling}
The family in Eq.~\eqref{eq:supp-collective-family} admits the certified bound
\begin{equation}
 q_{\rm col}
 \le
 \sum_{j=0}^{\lfloor n/2\rfloor}
 \left[q_{\rm univ}(n-2j+1)+1\right]-1
 =\mathcal O(n^3).
 \label{eq:supp-collective-cubic}
\end{equation}
\end{proposition}

\begin{proof}
Schur--Weyl duality gives~\smcite{3}
\begin{equation}
 (\mathbb C^2)^{\otimes n}
 \cong
 \bigoplus_{j=0}^{\lfloor n/2\rfloor}
 S^{(n-j,j)}\otimes V_{S_j},
 \qquad
 S_j=\frac n2-j,
 \label{eq:app-collective-sw}
\end{equation}
where $S^{(n-j,j)}$ is the permutation multiplicity space. Since each
collective operator acts as $I_{S^{(n-j,j)}}\otimes J_\mu^{(S_j)}$, the
Hamiltonian decomposes as
\begin{equation*}
 H_{\rm col}
 \cong
 \bigoplus_{j=0}^{\lfloor n/2\rfloor}
 I_{S^{(n-j,j)}}\otimes H_{\rm col}^{(S_j)},
 \qquad
 H_{\rm col}^{(S)}
 =\sum_{\mu=x,y,z}h_\mu J_\mu^{(S)}
 +\frac1n\sum_{\mu\le\nu}\chi_{\mu\nu}
 \left(J_\mu^{(S)}J_\nu^{(S)}+J_\nu^{(S)}J_\mu^{(S)}\right).
\end{equation*}
The multiplicity space therefore carries no Hamiltonian dependence. The active
spin space has dimension
\begin{equation}
 d_j=\dim V_{S_j}=2S_j+1=n-2j+1.
 \label{eq:app-collective-dimension}
\end{equation}
Applying the universal protocol of Ref.~\smcite{2} to each active
block and using
Corollary~\ref{cor:heterogeneous-fallback} yields the sum in
Eq.~\eqref{eq:supp-collective-cubic}. Since
\begin{equation}
 \sum_{j=0}^{\lfloor n/2\rfloor}(n-2j+1)^2=\mathcal O(n^3),
 \label{eq:supp-collective-square-sum}
\end{equation}
the broad family has the stated cubic certified budget. Treating the full
$2^n$-dimensional unitary as structure blind instead gives
$q_{\rm univ}(2^n)=\mathcal O(4^n)$.
\end{proof}

Quadratic collective terms need not be traceless on every active block. Their
block-dependent determinant phases are why the robust broad-family estimate
uses automatic completion. A further reduction for a parity-refined LMG
subfamily would require a separate trace-vector synchronization certificate.

\subsection{Passive multimode-link hierarchy}
\label{appendix:passive-link-scaling}

For $n$ physical modes, let
\begin{equation}
 H_{\rm link}=\sum_{r,s=1}^{n}X_{rs}a_r^\dagger a_s.
 \label{eq:supp-passive-H}
\end{equation}
Here $a_r$ is a mode annihilation operator and $X=(X_{rs})\in\mathbb C^{n\times n}$ is the unknown Hermitian single-particle coupling matrix. Thus $X_{rs}$ are scalar coefficients, while $H_{\rm link}$ acts on the Fock space. The total number
$\widehat N=\sum_r a_r^\dagger a_r$ is conserved.

\begin{proposition}[Shared single-particle inversion]
\label{prop:supp-passive-window}
Let $N\ge1$. On the at-most-$N$ excitation space,
\begin{equation}
 \mathcal H_{\le N}
 =\bigoplus_{r=0}^{N}\operatorname{Sym}^r(\mathbb C^n),
 \qquad
 d_r(n)=\binom{n+r-1}{r},
 \label{eq:supp-passive-sectors}
\end{equation}
Here $\operatorname{Sym}^r(\mathbb C^n)$ is the subspace of
$(\mathbb C^n)^{\otimes r}$ invariant under permutations of the $r$ tensor
factors, and $\operatorname{Sym}^r(e^{iX})$ denotes the restriction of
$(e^{iX})^{\otimes r}$ to this subspace. The forward evolution is
\begin{equation}
 U_{\le N}(X)
 =\bigoplus_{r=0}^{N}\operatorname{Sym}^r(e^{iX}).
 \label{eq:supp-passive-second-quantized-unitary}
\end{equation}
If the single-particle transformation $e^{iX}$ has a clean $q$-query inverter, then
$U_{\le N}(X)$ has a clean exact inverter using $N(q+1)$ forward calls to
$U_{\le N}(X)$. In particular,
\begin{equation}
 q_{\rm pass}
 \le N\left[q_{\rm univ}(n)+1\right]
 =\mathcal O(Nn^2).
 \label{eq:supp-passive-budget}
\end{equation}
Thus the bound is $\mathcal O(n^2)$ for fixed $N$, including $N=2$,
$\mathcal O(n^{5/2})$ for $N=\lceil\sqrt n\rceil$, and $\mathcal O(n^3)$ for $N=n$.
\end{proposition}

\begin{proof}
The $r$-excitation sector of $n$ bosonic modes is the symmetric power in
Eq.~\eqref{eq:supp-passive-sectors}; its occupation-number basis has
$\binom{n+r-1}{r}$ elements. The commutation relations give
\[
 [H_{\rm link},a_k^\dagger]
 =\sum_{s=1}^{n}X_{sk}a_s^\dagger,
 \qquad
 e^{iH_{\rm link}}a_k^\dagger e^{-iH_{\rm link}}
 =\sum_{s=1}^{n}(e^{iX})_{sk}a_s^\dagger.
\]
Since $H_{\rm link}$ annihilates the vacuum, applying this identity to a
product of $r$ creation operators gives
\[
 e^{iH_{\rm link}}a_{k_1}^\dagger\cdots a_{k_r}^\dagger|0\rangle
 =\prod_{j=1}^{r}
 \left[\sum_{s=1}^{n}(e^{iX})_{s k_j}a_s^\dagger\right]|0\rangle.
\]
Thus the restriction to the $r$-photon sector is
$\operatorname{Sym}^r(e^{iX})$, proving
Eq.~\eqref{eq:supp-passive-second-quantized-unitary}.

Let $\mathcal P[e^{iX}]=\rho(X)e^{-iX}$ be a clean single-particle inverter
using $q$ forward calls to $e^{iX}$. One such call can be implemented using one call to
$U_{\le N}(X)$: a fixed isometry loads the single-particle state into the
one-photon sector of a Fock register, the full evolution is applied once, and
the state is unloaded.

Consider first the branch with exactly $r$ photons. A fixed first-quantization
isometry identifies $\operatorname{Sym}^r(\mathbb C^n)$ with the symmetric
subspace of $(\mathbb C^n)^{\otimes r}$. Applying $\mathcal P[e^{iX}]$ to each
of the $r$ particle registers uses $rq$ calls. On the symmetric subspace, its
action is
\begin{equation}
 \left.\mathcal P[e^{iX}]^{\otimes r}
 \right|_{\operatorname{Sym}^r(\mathbb C^n)}
 =\rho(X)^r
 \left.(e^{-iX})^{\otimes r}
 \right|_{\operatorname{Sym}^r(\mathbb C^n)}
 =\rho(X)^r\operatorname{Sym}^r(e^{-iX}).
 \label{eq:supp-passive-particle-inversion}
\end{equation}
To make the phase and query number independent of $r$, run an additional copy
of $\mathcal P[e^{iX}]$ on an auxiliary single-particle register and follow it
with one forward call to $e^{iX}$. This $(q+1)$-call sequence acts as
\begin{equation}
 e^{iX}\mathcal P[e^{iX}]=\rho(X)I.
 \label{eq:supp-passive-scalar-closure}
\end{equation}
Running $N-r$ such closures changes the residual phase in
Eq.~\eqref{eq:supp-passive-particle-inversion} to $\rho(X)^N$. At this point
the branch has used
\begin{equation}
 rq+(N-r)(q+1)=N(q+1)-r
 \label{eq:supp-passive-branch-before-vacuum}
\end{equation}
calls. Finally, make $r$ additional calls with the queried Fock register in
the vacuum sector. Since $\operatorname{Sym}^0(e^{iX})=1$, these calls change
neither the state nor the phase. Every photon-number branch now uses exactly
$N(q+1)$ calls and implements
\begin{equation}
 \rho(X)^N\operatorname{Sym}^r(e^{-iX}).
 \label{eq:supp-passive-common-branch-output}
\end{equation}
Applying the inverse first-quantization isometry returns the state to the
$r$-photon Fock sector and resets its work registers.
The photon number can be written to a work register by a fixed unitary. The
same synchronized-branch construction as in
Lemma~\ref{lem:structured-controlled-sum} therefore combines the branches
coherently and then erases this register. All first-quantization, loading, and
routing operations are fixed, and every query slot contains the same ordinary
forward evolution $U_{\le N}(X)$. The resulting circuit is clean and has the
common global phase $\rho(X)^N$.

Taking $q=q_{\rm univ}(n)=\mathcal O(n^2)$ proves
Eq.~\eqref{eq:supp-passive-budget}. The three stated asymptotic regimes follow
by substituting the corresponding values of $N$.
\end{proof}

\begin{proposition}[Circulant and bright-mode certificates]
\label{prop:supp-passive-characters}
Restrict Eq.~\eqref{eq:supp-passive-H} to $\mathcal H_{\le2}$.

If $X$ is Hermitian circulant, the exact query number is $q_{\rm circ}=n$.

If $n\ge2$ and
\begin{equation}
 X=x_0I+x_1|b\rangle\!\langle b|,
 \qquad
 |b\rangle=\frac1{\sqrt n}\sum_{r=1}^{n}|r\rangle,
 \label{eq:supp-bright-X}
\end{equation}
with $x_0,x_1\in\RR$, then the exact query number on $\mathcal H_{\le2}$ is $q_{\rm bright}=2$,
independently of $n$. On the exactly-two-excitation sector, the exact query
number is one.
\end{proposition}

\begin{proof}
Every Hermitian circulant matrix is diagonalized by the same discrete Fourier
transform. The map from a Hermitian circulant matrix to its Fourier
eigenvalues is an invertible real-linear map onto \(\RR^n\), so these real
eigenvalues can vary independently. In the transformed modes $b_k$,
\begin{equation}
 H_{\rm link}=\sum_{k=1}^{n}\lambda_k(X)b_k^\dagger b_k.
 \label{eq:supp-circulant-diagonal}
\end{equation}
The fixed occupation basis therefore diagonalizes the entire hidden family.
Its character vectors on $\mathcal H_{\le2}$ form
\begin{equation}
 \Lambda_{\rm circ}
 =\left\{m\in\ZZ_{\ge0}^{n}:\lVert m\rVert_1\le2\right\}.
 \label{eq:supp-circulant-characters}
\end{equation}
The $q$-fold sumset is
\begin{equation}
 \Sigma_q(\Lambda_{\rm circ})
 =\left\{v\in\ZZ_{\ge0}^{n}:\lVert v\rVert_1\le2q\right\}.
 \label{eq:supp-circulant-sumset}
\end{equation}
For sufficiency in Corollary~\ref{cor:multiparameter-sumset}, choose
$c=2\mathbf 1_n$. Then
$c-\Lambda_{\rm circ}\subseteq\Sigma_n(\Lambda_{\rm circ})$, so $n$ queries
suffice. Conversely, because $2e_i\in\Lambda_{\rm circ}$, nonnegativity of
$c-2e_i$ requires $c_i\ge2$ for every $i$. The vacuum character also belongs
to $\Lambda_{\rm circ}$, so $c\in\Sigma_q(\Lambda_{\rm circ})$ and
$\lVert c\rVert_1\le2q$. Hence $2n\le\lVert c\rVert_1\le2q$, proving
$q\ge n$.

For Eq.~\eqref{eq:supp-bright-X}, the assumption $n\ge2$ leaves at least one
dark mode. Use the bright mode $b$ and any fixed basis of the orthogonal dark
manifold. A phase character is determined only by the
pair $(N,k_b)$, where $N$ is the total occupation and $k_b$ is the bright
occupation. On $N\le2$, the allowed pairs are
\begin{equation}
 (0,0),\quad (1,0),(1,1),\quad (2,0),(2,1),(2,2).
 \label{eq:supp-bright-pairs}
\end{equation}
Let this set be $\Lambda_{\rm bright}$. The choice $c=(4,2)$ satisfies
$c-\Lambda_{\rm bright}\subseteq\Sigma_2(\Lambda_{\rm bright})$, so two
queries suffice. One query cannot work: because $(0,0)\in\Lambda_{\rm bright}$,
it would require $c\in\Lambda_{\rm bright}$; because $(2,2)$ is also present
and $c-(2,2)$ must be nonnegative, this forces $c=(2,2)$, but then
$c-(2,0)=(0,2)\notin\Lambda_{\rm bright}$. Thus $q_{\rm bright}=2$.

On exactly $N=2$, the character set is
$\Lambda_2=\{(2,0),(2,1),(2,2)\}$. For the same $c=(4,2)$,
$c-\Lambda_2=\Lambda_2$, so one query suffices; zero queries are impossible
for three distinct characters. The growing dark-mode degeneracy changes only
multiplicity and does not enter either optimum.
\end{proof}

\section{Assisted phase certificates and catalogue synchronization}
\label{appendix:assisted-phase-certificates}

The automatic completion above starts from an inverse for each individual
block. A more specialized Hamiltonian family may provide something shorter:
one block may be inverted by using the full reduced operation, or a circuit
may return only a useful scalar phase. This section shows how any finite list
of such verified circuits can be combined. Whenever the resulting integer
program is feasible, it gives a certified upper bound. A finite list does not,
by itself, prove that the bound is optimal.

\subsection{Sector and scalar certificates}

Let exactness be required on a linear parameter domain \(L\subseteq\RR^p\).  In
the main text \(L=\RR^p\), but the quotient notation below makes linear
constraints harmless.  Define the charge space
\begin{equation}
    \mathfrak X_L
    :=
    (\RR^p)^*
    \big/
    \{r:r\cdot x=0\ \text{for all }x\in L\}.
    \label{eq:assisted-charge-space}
\end{equation}
A charge \(r\in\mathfrak X_L\) represents the character \(e^{ir\cdot x}\), and all charge
equalities below are equalities in \(\mathfrak X_L\).

\begin{definition}[Assisted sector certificate]
\label{def:assisted-sector-certificate}
For a block \(\alpha\), a query number \(\ell\), and a charge \(r\in\mathfrak X_L\), an
assisted sector certificate is a clean \(\ell\)-query circuit
\(\mathcal B_\alpha\) on \(\cH_{\mathrm{red}}\) such that
\begin{equation}
    \mathcal B_\alpha[U_{\mathrm{red}}](x)\iota^R_\alpha
    =
    e^{ir\cdot x}\,
    \iota^R_\alpha U_\alpha(x)^\dagger
    \qquad
    \forall x\in L .
    \label{eq:assisted-sector-certificate}
\end{equation}
\end{definition}

Only the input sector is specified. The branch may call the whole reduced
operation \(U_{\mathrm{red}}\), and it need not invert the other sectors. A local
block inverter imported by Lemma~\ref{lem:structured-embed-block} is the
special case in which the branch stays inside \(\iota^R_\alpha(\cV_\alpha)\).
Cross-block shortcuts are also included. If a fixed unitary isomorphism
\(T:\cV_\alpha\to \cV_\beta\), and hence $d_\alpha=d_\beta$, satisfies
\begin{equation}
    T^\dagger H_{j,\beta}T
    =
    a_j I_{\cV_\alpha}-H_{j,\alpha}
    \qquad
    \forall j,
\end{equation}
then
\begin{equation}
    T^\dagger U_\beta(x)T
    =
    e^{ia\cdot x}U_\alpha(x)^\dagger.
\end{equation}
Routing the \(\alpha\)-sector to \(\beta\), applying the reduced operation once, and
routing back gives a one-query assisted sector certificate with charge \(a\).

\begin{definition}[Scalar certificate]
\label{def:assisted-scalar-certificate}
For a query number \(w\) and a charge \(h\in\mathfrak X_L\), a scalar certificate is a
clean \(w\)-query circuit \(\mathcal S\) on \(\cH_{\mathrm{red}}\) such that
\begin{equation}
    \mathcal S[U_{\mathrm{red}}](x)
    =
    e^{ih\cdot x}I_{\cH_{\mathrm{red}}}
    \qquad
    \forall x\in L .
    \label{eq:assisted-scalar-certificate}
\end{equation}
\end{definition}

The determinant circuits are scalar certificates with query number
\(d_\beta\) and charge \(\tau_\beta\).  The closure lemma used for local
inverters has the same proof for assisted sector certificates.

\begin{lemma}[Closing an assisted sector certificate]
\label{lem:closing-assisted-sector}
If \((q,r)\) is certified for block \(\alpha\) by
Eq.~\eqref{eq:assisted-sector-certificate}, then \((q+1,r)\) is a scalar
certificate.
\end{lemma}

\begin{proof}
Bypass the data register. Prepare an auxiliary register in
\(\iota^R_\alpha(\cV_\alpha)\), run the assisted sector certificate on that
register, and then apply \(U_{\mathrm{red}}(x)\) once to the same register. On the legal
sector the auxiliary state undergoes
\[
    U_\alpha(x)\,e^{ir\cdot x}U_\alpha(x)^\dagger
    =
    e^{ir\cdot x}I_{\cV_\alpha}.
\]
The auxiliary register is restored, and the bypassed data register acquires the scalar
phase \(e^{ir\cdot x}\).  The cost is \(q+1\).
\end{proof}

\subsection{Character rigidity}

The catalogue ILP below synchronizes charges.  The next lemma explains why an
exact assisted sector certificate still has a character residual phase, even
though the branch may use all blocks.

\begin{lemma}[Character residual phase]
\label{lem:assisted-character-rigidity}
Let \(\mathcal B_\alpha\) be a clean \(q\)-query circuit on
\(\cH_{\mathrm{red}}\).  Suppose there is a phase function
\(\rho_\alpha:L\to U(1)\) such that
\begin{equation}
    \mathcal B_\alpha[U_{\mathrm{red}}](x)\iota^R_\alpha
    =
    \rho_\alpha(x)\,
    \iota^R_\alpha U_\alpha(x)^\dagger
    \qquad
    \forall x\in L .
    \label{eq:assisted-character-hypothesis}
\end{equation}
After multiplying the circuit by a fixed phase, there exists
\(r_\alpha\in\mathfrak X_L\) such that
\begin{equation}
    \rho_\alpha(x)=e^{ir_\alpha\cdot x}
    \qquad
    \forall x\in L .
\end{equation}
\end{lemma}

\begin{proof}
Fix \(c\in L\) and set \(x=tc\).  Write
\[
    H_{\mathrm{red}}(c)
    =
    \bigoplus_{\beta\in\Lambda_{\mathrm{blk}}}
    H_\beta(c),
    \qquad
    H_\beta(c)=\sum_j c_jH_{j,\beta}.
\]
Then \(U_{\mathrm{red}}(tc)=e^{itH_{\mathrm{red}}(c)}\). Expanding the
\(q\) forward calls in a spectral decomposition of
\(H_{\mathrm{red}}(c)\), every matrix element of
\(\mathcal B_\alpha[U_{\mathrm{red}}](tc)\) is a finite exponential
polynomial in \(t\). The phase can be extracted by
\begin{equation}
    \rho_\alpha(tc)
    =
    \frac{1}{d_\alpha}
    \operatorname{tr}_{\cV_\alpha}
    \left[
    (\iota^R_\alpha)^\dagger
    \mathcal B_\alpha[U_{\mathrm{red}}](tc)
    \iota^R_\alpha U_\alpha(tc)
    \right],
    \label{eq:assisted-phase-trace}
\end{equation}
so \(\rho_\alpha(tc)\) is also a finite exponential polynomial. Since its
modulus is one for every \(t\), it has only one Fourier mode.  Otherwise
\(|\rho_\alpha(tc)|^2\) would contain a nonzero highest frequency.  Thus, after
a fixed phase normalization,
\begin{equation}
    \rho_\alpha(tc)=e^{it r_\alpha(c)}
\end{equation}
for some real number \(r_\alpha(c)\).

It remains to see that \(r_\alpha(c)\) is linear in \(c\).  Differentiating
Eq.~\eqref{eq:assisted-character-hypothesis} at \(t=0\) gives, on the
\(\alpha\)-sector,
\[
    i\bigl(r_\alpha(c)I_{\cV_\alpha}-H_\alpha(c)\bigr).
\]
The same derivative computed from the clean circuit is real-linear in \(c\):
the product rule gives a sum over query slots, and each first-order term contains
exactly one insertion of \(iH_{\mathrm{red}}(c)\).  Taking the
normalized trace of the restricted first-order coefficient shows that
\(r_\alpha(c)\) is real-linear in \(c\).
This linear functional defines a charge in \(\mathfrak X_L\).
\end{proof}

\subsection{Catalogue synchronization}

Define the true sector and scalar frontiers by
\begin{align}
    R_\alpha(q)
    &:=
    \{r\in\mathfrak X_L:\text{there is a \(q\)-query assisted sector certificate for }
    \alpha\text{ with charge }r\},\\
    R_{\mathrm{sc}}(q)
    &:=
    \{h\in\mathfrak X_L:\text{there is a \(q\)-query scalar certificate with charge }h\}.
\end{align}
Appending a scalar certificate gives
\begin{equation}
    r\in R_\alpha(q),\quad h\in R_{\mathrm{sc}}(s)
    \quad\Longrightarrow\quad
    r+h\in R_\alpha(q+s).
    \label{eq:assisted-frontier-closure}
\end{equation}
The exact but abstract optimality statement is
\begin{equation}
    q_{\mathrm{opt}}
    =
    \min\left\{
    q:\bigcap_{\alpha\in\Lambda_{\mathrm{blk}}}R_\alpha(q)\ne\varnothing
    \right\}.
    \label{eq:assisted-true-frontier-opt}
\end{equation}
Indeed, a global inverse restricts to one sector certificate for each block.
Conversely, a common charge in all \(R_\alpha(q)\) gives same-length branches
with the same residual phase, which Lemma~\ref{lem:structured-controlled-sum}
combines coherently.

A finite certificate catalogue \(\cC\) consists of sector atoms
\[
    (\ell_{\alpha,k},r_{\alpha,k})
    \in \mathsf B_\alpha
\]
and scalar atoms
\[
    (w_\gamma,h_\gamma)\in\mathsf S_{\mathrm{atom}},
\]
each backed by an explicit clean circuit. The catalogue-generated frontier is
\begin{equation}
    R_\alpha^\cC(q)
    :=
    \left\{
    r_{\alpha,k}+\sum_\gamma n_{\alpha\gamma}h_\gamma:
    \ell_{\alpha,k}+\sum_\gamma n_{\alpha\gamma}w_\gamma=q,\ 
    n_{\alpha\gamma}\in\ZZ_{\ge0}
    \right\}.
    \label{eq:assisted-catalogue-frontier}
\end{equation}
By construction, \(R_\alpha^\cC(q)\subseteq R_\alpha(q)\).

Choose a basis of the finite-dimensional charge space $\mathfrak X_L$ and
represent all charges by their coordinates in that basis. The best protocol
generated by \(\cC\) is then the following mixed-integer linear program.
Minimize \(q\) over
\[
    z_{\alpha,k}\in\{0,1\},\qquad
    n_{\alpha\gamma}\in\ZZ_{\ge0},\qquad
    r_\star\in\mathfrak X_L,
\]
subject to, for every \(\alpha\),
\begin{align}
    \sum_k z_{\alpha,k}&=1,\\
    \sum_k z_{\alpha,k}\ell_{\alpha,k}
    +\sum_\gamma n_{\alpha\gamma}w_\gamma&=q,
    \label{eq:assisted-catalogue-ilp-cost}\\
    \sum_k z_{\alpha,k}r_{\alpha,k}
    +\sum_\gamma n_{\alpha\gamma}h_\gamma&=r_\star
    \qquad\text{in }\mathfrak X_L.
    \label{eq:assisted-catalogue-ilp-charge}
\end{align}
Equivalently,
\begin{equation}
    q(\cC)
    =
    \min\left\{
    q:\bigcap_\alpha R_\alpha^\cC(q)\ne\varnothing
    \right\}.
    \label{eq:assisted-catalogue-value}
\end{equation}
with the convention $q(\cC)=+\infty$ when the feasible set is empty.
Feasibility at \(q\) is exactly the condition needed by the controlled branch
compiler: on branch \(\alpha\), choose one sector atom and append the listed
scalar atoms.  Equations~\eqref{eq:assisted-catalogue-ilp-cost} and
\eqref{eq:assisted-catalogue-ilp-charge} make the branch lengths and charges
common. Hence every finite catalogue with $q(\cC)<+\infty$ gives a certified
upper bound,
\begin{equation}
    q_{\mathrm{opt}}\le q(\cC).
    \label{eq:assisted-catalogue-upper}
\end{equation}
The determinant trace-vector construction is the special case
\[
    (\ell_\alpha,r_\alpha)
    =
    (d_\alpha s_\alpha-1,\ s_\alpha\tau_\alpha),
    \qquad
    (w_\beta,h_\beta)
    =
    (d_\beta,\tau_\beta).
\]

\subsection{Hierarchy and finite optimality}

If
\(\cC_1\subseteq\cC_2\subseteq\cdots\) is an increasing sequence of finite
catalogues and \(q_m=q(\cC_m)\), then
\begin{equation}
    q_{\mathrm{opt}}\le q_{m+1}\le q_m.
    \label{eq:assisted-monotone-hierarchy}
\end{equation}
Adding a local inverter, a cross-block shortcut, a scalar gadget, or a
nonminimal but phase-useful sector primitive can only improve the catalogue
upper bound.

The certificate language is complete in the formal exhaustive limit.  Let
\(\mathfrak C\) be the directed collection of all finite certified catalogues.
Then
\begin{equation}
    q_{\mathrm{opt}}
    =
    \inf_{\cC\in\mathfrak C}q(\cC).
    \label{eq:assisted-exhaustive-completeness}
\end{equation}
The inequality \(q_{\mathrm{opt}}\le q(\cC)\) holds for every finite
catalogue. Conversely, take an optimal global reduced-system circuit
and restrict it to each active sector.  These restrictions are finitely many
assisted sector certificates, all with the same cost and charge, so they form a
finite catalogue whose ILP value is \(q_{\mathrm{opt}}\).

At a finite stage, optimality needs more information.  A catalogue is
frontier-complete if
\begin{equation}
    R_\alpha(q)=R_\alpha^\cC(q)
    \qquad
    \forall \alpha,\ q\ge0.
    \label{eq:assisted-frontier-complete}
\end{equation}
This condition implies \(q(\cC)=q_{\mathrm{opt}}\).  The weaker condition
\begin{equation}
    \bigcap_\alpha R_\alpha(q)
    =
    \bigcap_\alpha R_\alpha^\cC(q)
    \qquad
    \forall q<q(\cC)
    \label{eq:assisted-common-frontier-complete}
\end{equation}
is already sufficient, because it rules out any unlisted common phase below the
catalogue value.  Thus the finite ILP is an optimality theorem only after a
frontier-completeness or lower-bound argument has been supplied.

\begin{remark}
Outer relaxations can provide such lower-bound information.  If sets
\(\overline R_\alpha(q)\supseteq R_\alpha(q)\) have empty intersection below
some \(q_{\mathrm{lb}}\), then \(q_{\mathrm{lb}}\le q_{\mathrm{opt}}\).  When this
lower bound meets a catalogue upper bound, the catalogue protocol is globally
optimal.  This is a certification route, not an extra construction used in the
main examples.
\end{remark}

\section*{References}
\begingroup
\small
\newcounter{smref}
\begin{list}{[\arabic{smref}]}{%
  \usecounter{smref}
  \setlength{\leftmargin}{2.4em}
  \setlength{\labelwidth}{1.8em}
  \setlength{\labelsep}{0.6em}
  \setlength{\itemsep}{0.25em}
  \setlength{\parsep}{0pt}
  \setlength{\topsep}{0.5em}
}
\item P. Etingof, O. Golberg, S. Hensel, T. Liu, A. Schwendner,
D. Vaintrob, and E. Yudovina, \emph{Introduction to Representation Theory},
Student Mathematical Library, Vol.~59 (American Mathematical Society,
Providence, RI, 2011).

\item Y.-A. Chen, Y. Mo, Y. Liu, L. Zhang, and X. Wang,
arXiv:2403.04704 [quant-ph] (2024).

\item D. Bacon, I. L. Chuang, and A. W. Harrow,
Phys. Rev. Lett. \textbf{97}, 170502 (2006).
\end{list}
\endgroup

\end{document}